\documentclass[aps,pra,twocolumn,superscriptaddress,nofootinbib,floatfix,citeautoscript,10pt]{revtex4-2}

\usepackage{silence}
\usepackage[linesnumbered,ruled]{algorithm2e}
\SetArgSty{textnormal}

\usepackage{amsmath}
\usepackage{amssymb}
\usepackage{amsthm}
\usepackage{bm} 
\usepackage{bbm} 
\usepackage{braket}
\usepackage{dsfont}
\usepackage[colorlinks]{hyperref}
\hypersetup{
    colorlinks  = true,
    citecolor   = PineGreen,
    linkcolor   = MidnightBlue,
    urlcolor    = PineGreen
}
\usepackage{xurl}
\hypersetup{breaklinks=true}
\usepackage[bb=stix]{mathalpha} 
\usepackage{mathtools}
\usepackage{thmtools}
\usepackage[dvipsnames]{xcolor}
\usepackage{tikz}
\usepackage{etoolbox}
\apptocmd{\sloppy}{\hbadness 10000\relax}{}{}

\DeclarePairedDelimiter{\abs}{\lvert}{\rvert} 

\DeclareMathOperator{\Sym}{Sym} 
\DeclareMathOperator{\round}{round} 
\DeclareMathOperator{\dist}{dist} 

\newcommand*{\N}{\mathbb{N}}
\newcommand*{\Z}{\mathbb{Z}}
\newcommand*{\R}{\mathbb{R}}
\newcommand*{\C}{\mathbb{C}}
\newcommand*{\one}{\mathds{1}} 
\newcommand*{\bit}{\mathbb{b}} 
\newcommand*{\qubit}{\mathbb{q}} 
\newcommand*{\defcolon}{\,:\,} 
\newcommand*{\scriptin}{\raisebox{0.15ex}{$\scriptscriptstyle\in$}} 
\newcommand*{\expval}[3]{\bra{#1}\negthickspace#2\negthickspace\ket{#3}} 

\declaretheorem[style=plain]{theorem}

\declaretheorem[style=plain,sibling=theorem]{corollary}
\declaretheorem[style=plain,sibling=theorem]{lemma}
\declaretheorem[style=plain,sibling=theorem]{proposition}
\declaretheorem[style=definition,sibling=theorem]{definition}
\declaretheorem[style=plain]{conjecture}

\begin{document}

\title{Refuting the QAOA fixed-angle conjecture}

\author{Lennart \surname{Binkowski}}
\email{lennart.binkowski@itp.uni-hannover.de}
\affiliation{Institut für Theoretische Physik, Leibniz Universität Hannover}
\affiliation{Strangeworks, Austin, TX 78746, United States}

\begin{abstract}
    The fixed-angle conjecture for the quantum approximate optimisation algorithm~(QAOA) is shown to be false for depth-2 QAOA on 9-regular graphs, and proven to be correct for depth-1 on any regular graph, as well as for any depth on all 2-regular graphs.
\end{abstract}

\maketitle

\section{\label{section:Introduction}Introduction}

The quantum approximate optimisation algorithm~(QAOA)~\cite{Farhi2014AQuantumApproximateOptimizationAlgorithm} is arguably the best-studied (near-term) quantum algorithm for combinatorial optimisation.
In essence, the QAOA employs parameterised quantum circuits for constructing an ansatz class to optimise over.
This translates the task of finding a good quantum state to determining (a few) promising parameter values.
However, finding (globally) optimal parameter values typically amounts to higher-dimensional, non-convex optimisation, which is often an even more challenging task than solving the original problem~\cite{Bittel2021TrainingVariationalQuantumAlgorithmsIsNpHard}.
Global parameter optimisation is tractable only for small problem instances and few parameters.
Initiated by Zhou~\textit{et~al.}~\cite{Zhou2020QuantumApproximateOptimizationAlgorithmPerformanceMechanismAndImplementationOnNearTermDevices} and investigated further by Brand\~{a}o~\textit{et~al.}~\cite{Brandao2018ForFixedControlParametersTheQuantumApproximateOptimizationAlgorithmSObjectiveFunctionValueConcentratesForTypicalInstances} and Galda~\textit{et~al.}~\cite{Galda2021TransferabilityOfOptimalQaoaParametersBetweenRandomGraphs}, transferring optimal parameters from small to larger instances has become a computationally cheap, yet practically well-performing alternative.
QAOA parameter transfer has historically been studied on the MaxCut problem and especially on $3$-regular graphs, but more recent studies also investigate the same principle for the weighted MaxCut problem~\cite{Shaydulin2023ParameterTransferForQuantumApproximateOptimizationOfWeightedMaxcut}, the Maximum Independent Set problem~\cite{Xu2025QaoaParameterTransferabilityForMaximumIndependentSetUsingGraphAttentionNetworks}, and even between different problem classes~\cite{Nguyen2025CrossProblemParameterTransferInQuantumApproximateOptimizationAlgorithmAMachineLearningApproach}.

Most of these similarity results for QAOA parameter transfer are directionless: given a class of problem instances, it is not clear on which instance the fixed-parameter QAOA performs worst or best.
This question was first addressed by Wurtz and Love~\cite{Wurtz2021MaxcutQuantumApproximateOptimizationAlgorithmPerformanceGuaranteesFor}, again for the MaxCut problem on $3$-regular graphs.
Using the QAOA's locality and exhaustive subgraph analysis---as already proposed by Farhi~\textit{et~al.}~\cite{Farhi2014AQuantumApproximateOptimizationAlgorithm}---they found that, for depth-$2$, the QAOA performs worst on $3$-regular graphs without cycles shorter than $6$.
On all these graphs, the relevant subgraphs are pairs of perfect binary trees, fused together at their root vertices.
Their results established that on every $3$-regular graph, irrespective of its actual size, the QAOA with tree-optimal parameter values performs at least as well as on any graph with no small cycles.
For larger depths, the computation-heavy proofs quickly become infeasible to carry out.
Nevertheless, they conjectured that the observed ``performance guarantee'' of the tree subgraphs extends to larger depths and to all regular graphs.
Shortly after, Wurtz and Lykov~\cite{Wurtz2021FixedAngleConjecturesForTheQuantumApproximateOptimizationAlgorithmOnRegularMaxcutGraphs} provided additional numerical evidence for the ``fixed-angle'' conjecture for depths up to $p = 11$, noting that, if true, the QAOA for $p \geq 11$ would have a strictly better performance guarantee on $3$-regular graphs than the famous Goemans--Williamson algorithm~\cite{Goemans1995ImprovedApproximationAlgorithmsForMaximumCutAndSatisfiabilityProblemsUsingSemidefiniteProgramming}.

The fixed-angle conjecture is widely believed to be true on $3$-regular graphs, constituting a promising, parameter optimisation-free approach to large instance sizes~\cite{Augustino2024StrategiesForRunningTheQaoaAtHundredsOfQubits}.
However, as mentioned before, the computational requirements for Wurtz and Love's proof strategy quickly blow up with the QAOA depth $p$, simply because the larger $p$, the more subgraphs need to be considered.
Without some clever tricks or a fundamentally different proof strategy, the interesting regime of $p \geq 11$ will probably never be verifiable.
In contrast, even though numerically (or even analytically) tractable, the low-depth, but large-degree regime is quite unexplored.\footnote{
    This is arguably the case, because the $d=3$ case is the most interesting one, since the QAOA's performance is known to degrade with increasing $d$, so that any performance guarantees would only match the Goemans--Williamson bound at larger depths.
}
While Wurtz--Love and Wurtz--Lykov formulate the fixed-angle conjecture for all $d$-regular graphs, their analyses mostly restrict to the $3$-regular case.
In this work, I am exploring the validity of the fixed-angle conjecture on general $d$-regular graphs.
The three findings are that the fixed-angle conjecture is i) true for all $d$ as long as $p = 1$, ii) true for $d = 1, 2$ and all $p$, iii) false in general as there exists a counterexample at $d = 9$ with $p = 2$.
The positive results combine closed-form expressions for the parameterised QAOA performance at $p = 1$, as established by Wang~\textit{et~al.}~\cite{Wang2018QuantumApproximateOptimizationAlgorithmForMaxcutAFermionicView}, with elementary binomial and trigonometric inequalities, and, for the $d = 2$ case, with Marwaha's~\cite{Marwaha2026TheQaoaOnTheRingOfDisagrees} recently established connection between the QAOA and quantum signal processing~\cite{Low2017OptimalHamiltonianSimulationByQuantumSignalProcessing}.
The counterexample to the fixed-angle conjecture is the graph $K_{9, 9}$, the complete bipartite graph of degree $d = 9$.
The depth-$1$ QAOA's performance on $K_{n, m}$ has been previously upper-bounded by Bae~\textit{et~al.}~\cite{Bae2025ModifiedRecursiveQaoaForExactMaxCutSolutionsOnBipartiteGraphsClosingTheGapBeyondQaoaLimit}.
However, as the counterexample needs $p = 2$, I am instead using an exact Dicke state simulator for inferring the QAOA's performance for given parameter values.
In addition, Marwaha's~\cite{Marwaha2021LocalClassicalMaxCutAlgorithmOutperformsQaoaOnHighGirthRegularGraphs} closed-form expression for the depth-$2$ QAOA's performance on any $d$-regular graph without short cycles is used for studying tree-optimal parameters.

\section{\label{section:Preliminaries}Preliminaries}

In this section, I introduce all the notation and vocabulary necessary for rigorously stating and discussing the conjecture.
To this end, I begin with the formulation of the MaxCut problem.
An instance of the MaxCut problem is fully described by an undirected graph $G = (V, E)$, i.e., by a set $V$ of vertices and a set $E$ of undirected edges between vertices in $V$.
Undirected means that an edge in $E$ is described by a two-element subset $i j \coloneqq \{i, j\}$ with vertices $i, j \in V$, $i \neq j$.
The objective of MaxCut is to cut as many edges as possible by partitioning $V$ into two vertex subsets; an edge is considered cut if its endpoints lie in different subsets.
The usual binary encoding of this task uses $n \coloneqq \abs{V} \in \N \coloneqq \{1, 2, \ldots\}$ bits and associates each partition $(V_{1}, V_{2})$ of $V \cong [n] \coloneqq \{1, \ldots, n\}$ with a bitstring $\bm{x} \in \bit^{n}$ via
\begin{align*}
    x_{i} = \begin{cases}
        0, & \text{if } i \in V_{1} \\
        1, & \text{if } i \in V_{2}.
    \end{cases}
\end{align*}
That is, every bitstring is considered a feasible assignment; MaxCut is an unconstrained optimisation problem.
Under this encoding, the objective function $f : \bit^{n} \rightarrow \N_{0}$, measuring the number of cut edges, reads
\begin{align*}
    f(\bm{x}) \coloneqq \sum_{i j \scriptin E} x_{i} (1 - x_{j}) + (1 - x_{i}) x_{j}.
\end{align*}

With the standard Ising encoding, assigning one qubit per bit and replacing $2 x_{i}$ by $1 - \texttt{Z}_{i}$, MaxCut's objective function is represented by the diagonal $n$-qubit objective Hamiltonian
\begin{align*}
    H_{f} = \sum_{i j \scriptin E} \tfrac{1}{2} (\one - \texttt{Z}_{i} \texttt{Z}_{j}).
\end{align*}

Given a depth $p \in \N$, the QAOA prescribes a class of ansatz states $\ket{\bm{\beta}, \bm{\gamma}}$, parameterised by $2 p$ real parameters $\bm{\beta}, \bm{\gamma}$ and built from the $p$-fold alternation of ``phase separator'' $\exp(-i \gamma H_{f})$ and ``mixer'' $\exp(-i \beta B)$ applied to the uniform superposition $\ket{\bm{+}}$ of all computational basis states, where $B \coloneqq \sum_{i = 1}^{n} \texttt{X}_{i}$.
The optimisation now runs over the $2 p$ real parameters $\bm{\beta}, \bm{\gamma}$ instead of $2^{n}$ bitstrings, and is guided by the parameterised expectation values
\begin{align*}
    F(\bm{\beta}, \bm{\gamma}) &\coloneqq \expval{\bm{\beta}, \bm{\gamma}}{H_{f}}{\bm{\beta}, \bm{\gamma}} \\
    &= \sum_{i j \scriptin E} \tfrac{1}{2} \expval{\bm{\beta}, \bm{\gamma}}{\one - \texttt{Z}_{i} \texttt{Z}_{j}}{\bm{\beta}, \bm{\gamma}} \\
    &\eqqcolon \sum_{i j \scriptin E} f_{i j}(\bm{\beta}, \bm{\gamma}).
\end{align*}

Most importantly for the subsequent analysis, the QAOA is local.
Denote the graph distance of a vertex $v$ to an edge $i j$ by
\begin{align*}
    \dist(v, i j) \coloneqq \min\{\dist(v, i), \dist(v, j)\}.
\end{align*}
In the Heisenberg picture, $f_{i j}$ is the expectation value in $\ket{\bm{+}}$ of $\tfrac{1}{2} (\one - \texttt{Z}_{i} \texttt{Z}_{j})$, conjugated by all $p$ QAOA layers, starting with the last one.
Conjugation by a mixer layer leaves the support of an operator unchanged, whereas in a phase-separator layer all gates on edges without an endpoint in the current support commute with the operator and cancel; the remaining gates enlarge the support by at most one unit of distance.
Hence, counting backwards from $k = 1$, phase-separator layer $k$ acts on an operator supported on $\{v \defcolon \dist(v, i j) \leq k - 1\}$, and only the following edges contribute to $f_{i j}$.

\begin{definition}\label{definition:CausalCone}
    For $p \in \N$ and $i j \in E$, the \emph{causal cone} $G_{i j}^{p} = (V_{i j}^{p}, E_{i j}^{p})$ is the subgraph of $G$ with
    \begin{align*}
        V_{i j}^{p} &\coloneqq \{v \in V \defcolon \dist(v, i j) \leq p\}, \\
        E_{i j}^{p} &\coloneqq \{u v \in E \defcolon \min\{\dist(u, i j), \dist(v, i j)\} \leq p - 1\}.
    \end{align*}
\end{definition}

That is, $G_{i j}^{p}$ is the subgraph induced by all vertices within distance $p$ of $i$ or $j$, minus all edges between two vertices at distance exactly $p$.
Therefore, $f_{i j}$ can be evaluated on the (potentially) much smaller qubit system of $G_{i j}^{p}$, reducing the evaluation of $F$ to $\abs{E}$ individual evaluations.
Furthermore, if $G_{i j}^{p}$ and $G_{i' j'}^{p}$ are isomorphic via a map sending $i j$ to $i' j'$, then $f_{i j}(\bm{\beta}, \bm{\gamma}) = f_{i' j'}(\bm{\beta}, \bm{\gamma})$ for every $(\bm{\beta}, \bm{\gamma}) \in \R^{2 p}$.
Hence it even suffices to aggregate all $G_{i j}^{p}$ into isomorphism classes, evaluate $f_{i j}$ for one representative of each class, and multiply it by the size of the class.

An important property of $G$ which has an interesting interplay with the QAOA's locality is its \emph{girth}, defined to be the length of a shortest cycle in $G$.
If $G$ is acyclic, i.e., a forest, its girth is set to $\infty$.
If $G$ has girth $\geq 2 (p + 1)$, every $G_{i j}^{p}$ is a tree:
Shortest paths to $\{i, j\}$, together with $i j$, form a spanning tree of $G_{i j}^{p}$, and any further edge $x y$ of $G_{i j}^{p}$ would close a cycle in $G$ of length at most $\dist(x, i j) + \dist(y, i j) + 2 \leq (p - 1) + p + 2 = 2 p + 1$.
Note that the removed boundary edges are essential here, since an edge between two vertices at distance $p$ can close a cycle of length $2 (p + 1)$.
For the special case of $G$ being a $d$-regular graph with girth $\geq 2 (p + 1)$, this renders all $G_{i j}^{p}$ isomorphic to the same tree $T_{d}^{p}$, defined below and illustrated in \autoref{figure:TreeExamples}.

\begin{definition}\label{definition:TreeGraph}
    For $p, d \in \N$, let $T_{d}^{p}$ be the tree obtained by joining the roots of two perfect $(d - 1)$-ary trees of depth $p$ by an edge.
    Equivalently, $T_{d}^{p} \cong G_{i j}^{p}$ for any edge $i j$ of any $d$-regular graph $G$ of girth $\geq 2 (p + 1)$.
\end{definition}

\begin{figure*}[t]
    \centering
    \begin{minipage}[t]{0.48\linewidth}
        \centering
        \resizebox{\linewidth}{!}{%
            \begin{tikzpicture}
                \coordinate (i) at (-0.300,0.000);
                \coordinate (j) at (0.300,0.000);
                \coordinate (i0) at (-1.100,1.000);
                \coordinate (i00) at (-1.200,2.000);
                \coordinate (i000) at (-1.400,2.600);
                \coordinate (i001) at (-1.800,2.200);
                \coordinate (i01) at (-2.100,1.100);
                \coordinate (i010) at (-2.600,1.400);
                \coordinate (i011) at (-2.600,0.800);
                \coordinate (i1) at (-1.100,-1.000);
                \coordinate (i10) at (-2.100,-1.100);
                \coordinate (i100) at (-2.600,-0.800);
                \coordinate (i101) at (-2.600,-1.400);
                \coordinate (i11) at (-1.200,-2.000);
                \coordinate (i110) at (-1.800,-2.200);
                \coordinate (i111) at (-1.400,-2.600);
                \coordinate (j0) at (1.100,1.000);
                \coordinate (j00) at (1.200,2.000);
                \coordinate (j000) at (1.400,2.600);
                \coordinate (j001) at (1.800,2.200);
                \coordinate (j01) at (2.100,1.100);
                \coordinate (j010) at (2.600,1.400);
                \coordinate (j011) at (2.600,0.800);
                \coordinate (j1) at (1.100,-1.000);
                \coordinate (j10) at (2.100,-1.100);
                \coordinate (j100) at (2.600,-0.800);
                \coordinate (j101) at (2.600,-1.400);
                \coordinate (j11) at (1.200,-2.000);
                \coordinate (j110) at (1.800,-2.200);
                \coordinate (j111) at (1.400,-2.600);

                \draw[black, thick] (i) -- (i0);
                \draw[black!55, thick] (i0) -- (i00);
                \draw[black!35, thin] (i00) -- (i000);
                \draw[black!35, thin] (i00) -- (i001);
                \draw[black!55, thick] (i0) -- (i01);
                \draw[black!35, thin] (i01) -- (i010);
                \draw[black!35, thin] (i01) -- (i011);
                \draw[black, thick] (i) -- (i1);
                \draw[black!55, thick] (i1) -- (i10);
                \draw[black!35, thin] (i10) -- (i100);
                \draw[black!35, thin] (i10) -- (i101);
                \draw[black!55, thick] (i1) -- (i11);
                \draw[black!35, thin] (i11) -- (i110);
                \draw[black!35, thin] (i11) -- (i111);
                \draw[black, thick] (j) -- (j0);
                \draw[black!55, thick] (j0) -- (j00);
                \draw[black!35, thin] (j00) -- (j000);
                \draw[black!35, thin] (j00) -- (j001);
                \draw[black!55, thick] (j0) -- (j01);
                \draw[black!35, thin] (j01) -- (j010);
                \draw[black!35, thin] (j01) -- (j011);
                \draw[black, thick] (j) -- (j1);
                \draw[black!55, thick] (j1) -- (j10);
                \draw[black!35, thin] (j10) -- (j100);
                \draw[black!35, thin] (j10) -- (j101);
                \draw[black!55, thick] (j1) -- (j11);
                \draw[black!35, thin] (j11) -- (j110);
                \draw[black!35, thin] (j11) -- (j111);
                \draw[red!70!black, very thick] (i) -- (j);

                \draw[black!45, thin, dashed] (i000) -- (-1.600,3.000);
                \draw[black!45, thin, dashed] (i000) -- (-1.800,2.800);
                \draw[black!45, thin, dashed] (i001) -- (-2.000,2.600);
                \draw[black!45, thin, dashed] (i001) -- (-2.200,2.400);
                \draw[black!45, thin, dashed] (i010) -- (-3.050,1.550);
                \draw[black!45, thin, dashed] (i010) -- (-3.050,1.250);
                \draw[black!45, thin, dashed] (i011) -- (-3.050,0.950);
                \draw[black!45, thin, dashed] (i011) -- (-3.050,0.650);
                \draw[black!45, thin, dashed] (i100) -- (-3.050,-0.650);
                \draw[black!45, thin, dashed] (i100) -- (-3.050,-0.950);
                \draw[black!45, thin, dashed] (i101) -- (-3.050,-1.250);
                \draw[black!45, thin, dashed] (i101) -- (-3.050,-1.550);
                \draw[black!45, thin, dashed] (i110) -- (-2.200,-2.350);
                \draw[black!45, thin, dashed] (i110) -- (-2.000,-2.600);
                \draw[black!45, thin, dashed] (i111) -- (-1.800,-2.800);
                \draw[black!45, thin, dashed] (i111) -- (-1.600,-3.000);
                \draw[black!45, thin, dashed] (j000) -- (1.600,3.000);
                \draw[black!45, thin, dashed] (j000) -- (1.800,2.800);
                \draw[black!45, thin, dashed] (j001) -- (2.000,2.600);
                \draw[black!45, thin, dashed] (j001) -- (2.200,2.400);
                \draw[black!45, thin, dashed] (j010) -- (3.050,1.550);
                \draw[black!45, thin, dashed] (j010) -- (3.050,1.250);
                \draw[black!45, thin, dashed] (j011) -- (3.050,0.950);
                \draw[black!45, thin, dashed] (j011) -- (3.050,0.650);
                \draw[black!45, thin, dashed] (j100) -- (3.050,-0.650);
                \draw[black!45, thin, dashed] (j100) -- (3.050,-0.950);
                \draw[black!45, thin, dashed] (j101) -- (3.050,-1.250);
                \draw[black!45, thin, dashed] (j101) -- (3.050,-1.550);
                \draw[black!45, thin, dashed] (j110) -- (2.200,-2.350);
                \draw[black!45, thin, dashed] (j110) -- (2.000,-2.600);
                \draw[black!45, thin, dashed] (j111) -- (1.800,-2.800);
                \draw[black!45, thin, dashed] (j111) -- (1.600,-3.000);

                \fill[fill=blue!70!black] (i0) circle (2.2pt);
                \fill[fill=blue!70!black] (i1) circle (2.2pt);
                \fill[fill=blue!70!black] (j0) circle (2.2pt);
                \fill[fill=blue!70!black] (j1) circle (2.2pt);
                \fill[fill=blue!70!white] (i00) circle (1.8pt);
                \fill[fill=blue!70!white] (i01) circle (1.8pt);
                \fill[fill=blue!70!white] (i10) circle (1.8pt);
                \fill[fill=blue!70!white] (i11) circle (1.8pt);
                \fill[fill=blue!70!white] (j00) circle (1.8pt);
                \fill[fill=blue!70!white] (j01) circle (1.8pt);
                \fill[fill=blue!70!white] (j10) circle (1.8pt);
                \fill[fill=blue!70!white] (j11) circle (1.8pt);
                \fill[fill=black!45] (i000) circle (1.5pt);
                \fill[fill=black!45] (i001) circle (1.5pt);
                \fill[fill=black!45] (i010) circle (1.5pt);
                \fill[fill=black!45] (i011) circle (1.5pt);
                \fill[fill=black!45] (i100) circle (1.5pt);
                \fill[fill=black!45] (i101) circle (1.5pt);
                \fill[fill=black!45] (i110) circle (1.5pt);
                \fill[fill=black!45] (i111) circle (1.5pt);
                \fill[fill=black!45] (j000) circle (1.5pt);
                \fill[fill=black!45] (j001) circle (1.5pt);
                \fill[fill=black!45] (j010) circle (1.5pt);
                \fill[fill=black!45] (j011) circle (1.5pt);
                \fill[fill=black!45] (j100) circle (1.5pt);
                \fill[fill=black!45] (j101) circle (1.5pt);
                \fill[fill=black!45] (j110) circle (1.5pt);
                \fill[fill=black!45] (j111) circle (1.5pt);
                \fill[fill=blue!70!black] (i) circle (2.4pt);
                \fill[fill=blue!70!black] (j) circle (2.4pt);
                \node[below=2pt] at (i) {\scriptsize $i$};
                \node[below=2pt] at (j) {\scriptsize $j$};
            \end{tikzpicture}%
        }
    \end{minipage}\hfill
    \begin{minipage}[t]{0.48\linewidth}
        \centering
        \resizebox{\linewidth}{!}{%
            \begin{tikzpicture}
                \coordinate (i) at (-0.300,0.000);
                \coordinate (j) at (0.300,0.000);
                \coordinate (i0) at (-0.800,1.000);
                \coordinate (i00) at (-0.650,1.750);
                \coordinate (i000) at (-0.410,2.100);
                \coordinate (i001) at (-0.565,2.175);
                \coordinate (i002) at (-0.735,2.150);
                \coordinate (i01) at (-1.130,1.700);
                \coordinate (i010) at (-1.190,2.160);
                \coordinate (i011) at (-1.373,2.147);
                \coordinate (i012) at (-1.490,2.010);
                \coordinate (i02) at (-1.480,1.400);
                \coordinate (i020) at (-1.740,1.770);
                \coordinate (i021) at (-1.900,1.640);
                \coordinate (i022) at (-1.950,1.440);
                \coordinate (i1) at (-1.400,0.000);
                \coordinate (i10) at (-1.950,0.600);
                \coordinate (i100) at (-2.050,1.000);
                \coordinate (i101) at (-2.280,0.930);
                \coordinate (i102) at (-2.330,0.700);
                \coordinate (i11) at (-2.200,0.000);
                \coordinate (i110) at (-2.600,0.250);
                \coordinate (i111) at (-2.650,0.000);
                \coordinate (i112) at (-2.600,-0.250);
                \coordinate (i12) at (-1.950,-0.600);
                \coordinate (i120) at (-2.330,-0.700);
                \coordinate (i121) at (-2.280,-0.930);
                \coordinate (i122) at (-2.050,-1.000);
                \coordinate (i2) at (-0.800,-1.000);
                \coordinate (i20) at (-0.650,-1.750);
                \coordinate (i200) at (-0.410,-2.100);
                \coordinate (i201) at (-0.565,-2.175);
                \coordinate (i202) at (-0.735,-2.150);
                \coordinate (i21) at (-1.130,-1.700);
                \coordinate (i210) at (-1.190,-2.160);
                \coordinate (i211) at (-1.373,-2.147);
                \coordinate (i212) at (-1.490,-2.010);
                \coordinate (i22) at (-1.480,-1.400);
                \coordinate (i220) at (-1.740,-1.770);
                \coordinate (i221) at (-1.900,-1.640);
                \coordinate (i222) at (-1.950,-1.440);
                \coordinate (j0) at (0.800,1.000);
                \coordinate (j00) at (0.650,1.750);
                \coordinate (j000) at (0.410,2.100);
                \coordinate (j001) at (0.565,2.175);
                \coordinate (j002) at (0.735,2.150);
                \coordinate (j01) at (1.130,1.700);
                \coordinate (j010) at (1.190,2.160);
                \coordinate (j011) at (1.373,2.147);
                \coordinate (j012) at (1.490,2.010);
                \coordinate (j02) at (1.480,1.400);
                \coordinate (j020) at (1.740,1.770);
                \coordinate (j021) at (1.900,1.640);
                \coordinate (j022) at (1.950,1.440);
                \coordinate (j1) at (1.400,0.000);
                \coordinate (j10) at (1.950,0.600);
                \coordinate (j100) at (2.050,1.000);
                \coordinate (j101) at (2.280,0.930);
                \coordinate (j102) at (2.330,0.700);
                \coordinate (j11) at (2.200,0.000);
                \coordinate (j110) at (2.600,0.250);
                \coordinate (j111) at (2.650,0.000);
                \coordinate (j112) at (2.600,-0.250);
                \coordinate (j12) at (1.950,-0.600);
                \coordinate (j120) at (2.330,-0.700);
                \coordinate (j121) at (2.280,-0.930);
                \coordinate (j122) at (2.050,-1.000);
                \coordinate (j2) at (0.800,-1.000);
                \coordinate (j20) at (0.650,-1.750);
                \coordinate (j200) at (0.410,-2.100);
                \coordinate (j201) at (0.565,-2.175);
                \coordinate (j202) at (0.735,-2.150);
                \coordinate (j21) at (1.130,-1.700);
                \coordinate (j210) at (1.190,-2.160);
                \coordinate (j211) at (1.373,-2.147);
                \coordinate (j212) at (1.490,-2.010);
                \coordinate (j22) at (1.480,-1.400);
                \coordinate (j220) at (1.740,-1.770);
                \coordinate (j221) at (1.900,-1.640);
                \coordinate (j222) at (1.950,-1.440);

                \draw[black, thick] (i) -- (i0);
                \draw[black!55, thick] (i0) -- (i00);
                \draw[black!35, thin] (i00) -- (i000);
                \draw[black!35, thin] (i00) -- (i001);
                \draw[black!35, thin] (i00) -- (i002);
                \draw[black!55, thick] (i0) -- (i01);
                \draw[black!35, thin] (i01) -- (i010);
                \draw[black!35, thin] (i01) -- (i011);
                \draw[black!35, thin] (i01) -- (i012);
                \draw[black!55, thick] (i0) -- (i02);
                \draw[black!35, thin] (i02) -- (i020);
                \draw[black!35, thin] (i02) -- (i021);
                \draw[black!35, thin] (i02) -- (i022);
                \draw[black, thick] (i) -- (i1);
                \draw[black!55, thick] (i1) -- (i10);
                \draw[black!35, thin] (i10) -- (i100);
                \draw[black!35, thin] (i10) -- (i101);
                \draw[black!35, thin] (i10) -- (i102);
                \draw[black!55, thick] (i1) -- (i11);
                \draw[black!35, thin] (i11) -- (i110);
                \draw[black!35, thin] (i11) -- (i111);
                \draw[black!35, thin] (i11) -- (i112);
                \draw[black!55, thick] (i1) -- (i12);
                \draw[black!35, thin] (i12) -- (i120);
                \draw[black!35, thin] (i12) -- (i121);
                \draw[black!35, thin] (i12) -- (i122);
                \draw[black, thick] (i) -- (i2);
                \draw[black!55, thick] (i2) -- (i20);
                \draw[black!35, thin] (i20) -- (i200);
                \draw[black!35, thin] (i20) -- (i201);
                \draw[black!35, thin] (i20) -- (i202);
                \draw[black!55, thick] (i2) -- (i21);
                \draw[black!35, thin] (i21) -- (i210);
                \draw[black!35, thin] (i21) -- (i211);
                \draw[black!35, thin] (i21) -- (i212);
                \draw[black!55, thick] (i2) -- (i22);
                \draw[black!35, thin] (i22) -- (i220);
                \draw[black!35, thin] (i22) -- (i221);
                \draw[black!35, thin] (i22) -- (i222);
                \draw[black, thick] (j) -- (j0);
                \draw[black!55, thick] (j0) -- (j00);
                \draw[black!35, thin] (j00) -- (j000);
                \draw[black!35, thin] (j00) -- (j001);
                \draw[black!35, thin] (j00) -- (j002);
                \draw[black!55, thick] (j0) -- (j01);
                \draw[black!35, thin] (j01) -- (j010);
                \draw[black!35, thin] (j01) -- (j011);
                \draw[black!35, thin] (j01) -- (j012);
                \draw[black!55, thick] (j0) -- (j02);
                \draw[black!35, thin] (j02) -- (j020);
                \draw[black!35, thin] (j02) -- (j021);
                \draw[black!35, thin] (j02) -- (j022);
                \draw[black, thick] (j) -- (j1);
                \draw[black!55, thick] (j1) -- (j10);
                \draw[black!35, thin] (j10) -- (j100);
                \draw[black!35, thin] (j10) -- (j101);
                \draw[black!35, thin] (j10) -- (j102);
                \draw[black!55, thick] (j1) -- (j11);
                \draw[black!35, thin] (j11) -- (j110);
                \draw[black!35, thin] (j11) -- (j111);
                \draw[black!35, thin] (j11) -- (j112);
                \draw[black!55, thick] (j1) -- (j12);
                \draw[black!35, thin] (j12) -- (j120);
                \draw[black!35, thin] (j12) -- (j121);
                \draw[black!35, thin] (j12) -- (j122);
                \draw[black, thick] (j) -- (j2);
                \draw[black!55, thick] (j2) -- (j20);
                \draw[black!35, thin] (j20) -- (j200);
                \draw[black!35, thin] (j20) -- (j201);
                \draw[black!35, thin] (j20) -- (j202);
                \draw[black!55, thick] (j2) -- (j21);
                \draw[black!35, thin] (j21) -- (j210);
                \draw[black!35, thin] (j21) -- (j211);
                \draw[black!35, thin] (j21) -- (j212);
                \draw[black!55, thick] (j2) -- (j22);
                \draw[black!35, thin] (j22) -- (j220);
                \draw[black!35, thin] (j22) -- (j221);
                \draw[black!35, thin] (j22) -- (j222);
                \draw[red!70!black, very thick] (i) -- (j);

                \draw[black!45, thin, dashed] (i000) -- (-0.170,2.300);
                \draw[black!45, thin, dashed] (i000) -- (-0.230,2.363);
                \draw[black!45, thin, dashed] (i000) -- (-0.300,2.400);
                \draw[black!45, thin, dashed] (i001) -- (-0.410,2.460);
                \draw[black!45, thin, dashed] (i001) -- (-0.500,2.500);
                \draw[black!45, thin, dashed] (i001) -- (-0.605,2.550);
                \draw[black!45, thin, dashed] (i002) -- (-0.725,2.455);
                \draw[black!45, thin, dashed] (i002) -- (-0.820,2.450);
                \draw[black!45, thin, dashed] (i002) -- (-0.900,2.400);
                \draw[black!45, thin, dashed] (i010) -- (-1.125,2.460);
                \draw[black!45, thin, dashed] (i010) -- (-1.210,2.480);
                \draw[black!45, thin, dashed] (i010) -- (-1.290,2.450);
                \draw[black!45, thin, dashed] (i011) -- (-1.430,2.447);
                \draw[black!45, thin, dashed] (i011) -- (-1.510,2.418);
                \draw[black!45, thin, dashed] (i011) -- (-1.580,2.367);
                \draw[black!45, thin, dashed] (i012) -- (-1.670,2.260);
                \draw[black!45, thin, dashed] (i012) -- (-1.725,2.210);
                \draw[black!45, thin, dashed] (i012) -- (-1.765,2.145);
                \draw[black!45, thin, dashed] (i020) -- (-1.890,2.150);
                \draw[black!45, thin, dashed] (i020) -- (-1.970,2.090);
                \draw[black!45, thin, dashed] (i020) -- (-2.020,2.000);
                \draw[black!45, thin, dashed] (i021) -- (-2.180,1.920);
                \draw[black!45, thin, dashed] (i021) -- (-2.250,1.840);
                \draw[black!45, thin, dashed] (i021) -- (-2.300,1.750);
                \draw[black!45, thin, dashed] (i022) -- (-2.250,1.540);
                \draw[black!45, thin, dashed] (i022) -- (-2.310,1.470);
                \draw[black!45, thin, dashed] (i022) -- (-2.360,1.380);
                \draw[black!45, thin, dashed] (i100) -- (-2.050,1.400);
                \draw[black!45, thin, dashed] (i100) -- (-2.125,1.300);
                \draw[black!45, thin, dashed] (i100) -- (-2.200,1.275);
                \draw[black!45, thin, dashed] (i101) -- (-2.430,1.190);
                \draw[black!45, thin, dashed] (i101) -- (-2.510,1.160);
                \draw[black!45, thin, dashed] (i101) -- (-2.540,1.080);
                \draw[black!45, thin, dashed] (i102) -- (-2.610,0.865);
                \draw[black!45, thin, dashed] (i102) -- (-2.650,0.785);
                \draw[black!45, thin, dashed] (i102) -- (-2.725,0.705);
                \draw[black!45, thin, dashed] (i110) -- (-2.800,0.525);
                \draw[black!45, thin, dashed] (i110) -- (-2.850,0.425);
                \draw[black!45, thin, dashed] (i110) -- (-2.900,0.325);
                \draw[black!45, thin, dashed] (i111) -- (-2.950,0.100);
                \draw[black!45, thin, dashed] (i111) -- (-3.050,0.000);
                \draw[black!45, thin, dashed] (i111) -- (-2.950,-0.100);
                \draw[black!45, thin, dashed] (i112) -- (-2.900,-0.325);
                \draw[black!45, thin, dashed] (i112) -- (-2.850,-0.425);
                \draw[black!45, thin, dashed] (i112) -- (-2.800,-0.525);
                \draw[black!45, thin, dashed] (i120) -- (-2.725,-0.705);
                \draw[black!45, thin, dashed] (i120) -- (-2.650,-0.785);
                \draw[black!45, thin, dashed] (i120) -- (-2.610,-0.865);
                \draw[black!45, thin, dashed] (i121) -- (-2.540,-1.080);
                \draw[black!45, thin, dashed] (i121) -- (-2.510,-1.160);
                \draw[black!45, thin, dashed] (i121) -- (-2.430,-1.190);
                \draw[black!45, thin, dashed] (i122) -- (-2.185,-1.275);
                \draw[black!45, thin, dashed] (i122) -- (-2.125,-1.300);
                \draw[black!45, thin, dashed] (i122) -- (-2.050,-1.400);
                \draw[black!45, thin, dashed] (i200) -- (-0.300,-2.400);
                \draw[black!45, thin, dashed] (i200) -- (-0.230,-2.363);
                \draw[black!45, thin, dashed] (i200) -- (-0.170,-2.300);
                \draw[black!45, thin, dashed] (i201) -- (-0.605,-2.550);
                \draw[black!45, thin, dashed] (i201) -- (-0.500,-2.500);
                \draw[black!45, thin, dashed] (i201) -- (-0.410,-2.460);
                \draw[black!45, thin, dashed] (i202) -- (-0.900,-2.400);
                \draw[black!45, thin, dashed] (i202) -- (-0.820,-2.450);
                \draw[black!45, thin, dashed] (i202) -- (-0.725,-2.455);
                \draw[black!45, thin, dashed] (i210) -- (-1.290,-2.450);
                \draw[black!45, thin, dashed] (i210) -- (-1.210,-2.480);
                \draw[black!45, thin, dashed] (i210) -- (-1.125,-2.460);
                \draw[black!45, thin, dashed] (i211) -- (-1.580,-2.367);
                \draw[black!45, thin, dashed] (i211) -- (-1.510,-2.418);
                \draw[black!45, thin, dashed] (i211) -- (-1.430,-2.447);
                \draw[black!45, thin, dashed] (i212) -- (-1.765,-2.145);
                \draw[black!45, thin, dashed] (i212) -- (-1.725,-2.210);
                \draw[black!45, thin, dashed] (i212) -- (-1.670,-2.260);
                \draw[black!45, thin, dashed] (i220) -- (-2.020,-2.000);
                \draw[black!45, thin, dashed] (i220) -- (-1.970,-2.090);
                \draw[black!45, thin, dashed] (i220) -- (-1.890,-2.150);
                \draw[black!45, thin, dashed] (i221) -- (-2.300,-1.750);
                \draw[black!45, thin, dashed] (i221) -- (-2.250,-1.840);
                \draw[black!45, thin, dashed] (i221) -- (-2.180,-1.920);
                \draw[black!45, thin, dashed] (i222) -- (-2.360,-1.380);
                \draw[black!45, thin, dashed] (i222) -- (-2.310,-1.470);
                \draw[black!45, thin, dashed] (i222) -- (-2.250,-1.540);
                \draw[black!45, thin, dashed] (j000) -- (0.170,2.300);
                \draw[black!45, thin, dashed] (j000) -- (0.230,2.363);
                \draw[black!45, thin, dashed] (j000) -- (0.300,2.400);
                \draw[black!45, thin, dashed] (j001) -- (0.410,2.460);
                \draw[black!45, thin, dashed] (j001) -- (0.500,2.500);
                \draw[black!45, thin, dashed] (j001) -- (0.605,2.550);
                \draw[black!45, thin, dashed] (j002) -- (0.725,2.455);
                \draw[black!45, thin, dashed] (j002) -- (0.820,2.450);
                \draw[black!45, thin, dashed] (j002) -- (0.900,2.400);
                \draw[black!45, thin, dashed] (j010) -- (1.125,2.460);
                \draw[black!45, thin, dashed] (j010) -- (1.210,2.480);
                \draw[black!45, thin, dashed] (j010) -- (1.290,2.450);
                \draw[black!45, thin, dashed] (j011) -- (1.430,2.447);
                \draw[black!45, thin, dashed] (j011) -- (1.510,2.418);
                \draw[black!45, thin, dashed] (j011) -- (1.580,2.367);
                \draw[black!45, thin, dashed] (j012) -- (1.670,2.260);
                \draw[black!45, thin, dashed] (j012) -- (1.725,2.210);
                \draw[black!45, thin, dashed] (j012) -- (1.765,2.145);
                \draw[black!45, thin, dashed] (j020) -- (1.890,2.150);
                \draw[black!45, thin, dashed] (j020) -- (1.970,2.090);
                \draw[black!45, thin, dashed] (j020) -- (2.020,2.000);
                \draw[black!45, thin, dashed] (j021) -- (2.180,1.920);
                \draw[black!45, thin, dashed] (j021) -- (2.250,1.840);
                \draw[black!45, thin, dashed] (j021) -- (2.300,1.750);
                \draw[black!45, thin, dashed] (j022) -- (2.250,1.540);
                \draw[black!45, thin, dashed] (j022) -- (2.310,1.470);
                \draw[black!45, thin, dashed] (j022) -- (2.360,1.380);
                \draw[black!45, thin, dashed] (j100) -- (2.050,1.400);
                \draw[black!45, thin, dashed] (j100) -- (2.125,1.300);
                \draw[black!45, thin, dashed] (j100) -- (2.200,1.275);
                \draw[black!45, thin, dashed] (j101) -- (2.430,1.190);
                \draw[black!45, thin, dashed] (j101) -- (2.510,1.160);
                \draw[black!45, thin, dashed] (j101) -- (2.540,1.080);
                \draw[black!45, thin, dashed] (j102) -- (2.610,0.865);
                \draw[black!45, thin, dashed] (j102) -- (2.650,0.785);
                \draw[black!45, thin, dashed] (j102) -- (2.725,0.705);
                \draw[black!45, thin, dashed] (j110) -- (2.800,0.525);
                \draw[black!45, thin, dashed] (j110) -- (2.850,0.425);
                \draw[black!45, thin, dashed] (j110) -- (2.900,0.325);
                \draw[black!45, thin, dashed] (j111) -- (2.950,0.100);
                \draw[black!45, thin, dashed] (j111) -- (3.050,0.000);
                \draw[black!45, thin, dashed] (j111) -- (2.950,-0.100);
                \draw[black!45, thin, dashed] (j112) -- (2.900,-0.325);
                \draw[black!45, thin, dashed] (j112) -- (2.850,-0.425);
                \draw[black!45, thin, dashed] (j112) -- (2.800,-0.525);
                \draw[black!45, thin, dashed] (j120) -- (2.725,-0.705);
                \draw[black!45, thin, dashed] (j120) -- (2.650,-0.785);
                \draw[black!45, thin, dashed] (j120) -- (2.610,-0.865);
                \draw[black!45, thin, dashed] (j121) -- (2.540,-1.080);
                \draw[black!45, thin, dashed] (j121) -- (2.510,-1.160);
                \draw[black!45, thin, dashed] (j121) -- (2.430,-1.190);
                \draw[black!45, thin, dashed] (j122) -- (2.185,-1.275);
                \draw[black!45, thin, dashed] (j122) -- (2.125,-1.300);
                \draw[black!45, thin, dashed] (j122) -- (2.050,-1.400);
                \draw[black!45, thin, dashed] (j200) -- (0.300,-2.400);
                \draw[black!45, thin, dashed] (j200) -- (0.230,-2.363);
                \draw[black!45, thin, dashed] (j200) -- (0.170,-2.300);
                \draw[black!45, thin, dashed] (j201) -- (0.605,-2.550);
                \draw[black!45, thin, dashed] (j201) -- (0.500,-2.500);
                \draw[black!45, thin, dashed] (j201) -- (0.410,-2.460);
                \draw[black!45, thin, dashed] (j202) -- (0.900,-2.400);
                \draw[black!45, thin, dashed] (j202) -- (0.820,-2.450);
                \draw[black!45, thin, dashed] (j202) -- (0.725,-2.455);
                \draw[black!45, thin, dashed] (j210) -- (1.290,-2.450);
                \draw[black!45, thin, dashed] (j210) -- (1.210,-2.480);
                \draw[black!45, thin, dashed] (j210) -- (1.125,-2.460);
                \draw[black!45, thin, dashed] (j211) -- (1.580,-2.367);
                \draw[black!45, thin, dashed] (j211) -- (1.510,-2.418);
                \draw[black!45, thin, dashed] (j211) -- (1.430,-2.447);
                \draw[black!45, thin, dashed] (j212) -- (1.765,-2.145);
                \draw[black!45, thin, dashed] (j212) -- (1.725,-2.210);
                \draw[black!45, thin, dashed] (j212) -- (1.670,-2.260);
                \draw[black!45, thin, dashed] (j220) -- (2.020,-2.000);
                \draw[black!45, thin, dashed] (j220) -- (1.970,-2.090);
                \draw[black!45, thin, dashed] (j220) -- (1.890,-2.150);
                \draw[black!45, thin, dashed] (j221) -- (2.300,-1.750);
                \draw[black!45, thin, dashed] (j221) -- (2.250,-1.840);
                \draw[black!45, thin, dashed] (j221) -- (2.180,-1.920);
                \draw[black!45, thin, dashed] (j222) -- (2.360,-1.380);
                \draw[black!45, thin, dashed] (j222) -- (2.310,-1.470);
                \draw[black!45, thin, dashed] (j222) -- (2.250,-1.540);

                \fill[fill=blue!70!black] (i0) circle (2.2pt);
                \fill[fill=blue!70!black] (i1) circle (2.2pt);
                \fill[fill=blue!70!black] (i2) circle (2.2pt);
                \fill[fill=blue!70!black] (j0) circle (2.2pt);
                \fill[fill=blue!70!black] (j1) circle (2.2pt);
                \fill[fill=blue!70!black] (j2) circle (2.2pt);
                \fill[fill=blue!70!white] (i00) circle (1.8pt);
                \fill[fill=blue!70!white] (i01) circle (1.8pt);
                \fill[fill=blue!70!white] (i02) circle (1.8pt);
                \fill[fill=blue!70!white] (i10) circle (1.8pt);
                \fill[fill=blue!70!white] (i11) circle (1.8pt);
                \fill[fill=blue!70!white] (i12) circle (1.8pt);
                \fill[fill=blue!70!white] (i20) circle (1.8pt);
                \fill[fill=blue!70!white] (i21) circle (1.8pt);
                \fill[fill=blue!70!white] (i22) circle (1.8pt);
                \fill[fill=blue!70!white] (j00) circle (1.8pt);
                \fill[fill=blue!70!white] (j01) circle (1.8pt);
                \fill[fill=blue!70!white] (j02) circle (1.8pt);
                \fill[fill=blue!70!white] (j10) circle (1.8pt);
                \fill[fill=blue!70!white] (j11) circle (1.8pt);
                \fill[fill=blue!70!white] (j12) circle (1.8pt);
                \fill[fill=blue!70!white] (j20) circle (1.8pt);
                \fill[fill=blue!70!white] (j21) circle (1.8pt);
                \fill[fill=blue!70!white] (j22) circle (1.8pt);
                \fill[fill=black!45] (i000) circle (1.5pt);
                \fill[fill=black!45] (i001) circle (1.5pt);
                \fill[fill=black!45] (i002) circle (1.5pt);
                \fill[fill=black!45] (i010) circle (1.5pt);
                \fill[fill=black!45] (i011) circle (1.5pt);
                \fill[fill=black!45] (i012) circle (1.5pt);
                \fill[fill=black!45] (i020) circle (1.5pt);
                \fill[fill=black!45] (i021) circle (1.5pt);
                \fill[fill=black!45] (i022) circle (1.5pt);
                \fill[fill=black!45] (i100) circle (1.5pt);
                \fill[fill=black!45] (i101) circle (1.5pt);
                \fill[fill=black!45] (i102) circle (1.5pt);
                \fill[fill=black!45] (i110) circle (1.5pt);
                \fill[fill=black!45] (i111) circle (1.5pt);
                \fill[fill=black!45] (i112) circle (1.5pt);
                \fill[fill=black!45] (i120) circle (1.5pt);
                \fill[fill=black!45] (i121) circle (1.5pt);
                \fill[fill=black!45] (i122) circle (1.5pt);
                \fill[fill=black!45] (i200) circle (1.5pt);
                \fill[fill=black!45] (i201) circle (1.5pt);
                \fill[fill=black!45] (i202) circle (1.5pt);
                \fill[fill=black!45] (i210) circle (1.5pt);
                \fill[fill=black!45] (i211) circle (1.5pt);
                \fill[fill=black!45] (i212) circle (1.5pt);
                \fill[fill=black!45] (i220) circle (1.5pt);
                \fill[fill=black!45] (i221) circle (1.5pt);
                \fill[fill=black!45] (i222) circle (1.5pt);
                \fill[fill=black!45] (j000) circle (1.5pt);
                \fill[fill=black!45] (j001) circle (1.5pt);
                \fill[fill=black!45] (j002) circle (1.5pt);
                \fill[fill=black!45] (j010) circle (1.5pt);
                \fill[fill=black!45] (j011) circle (1.5pt);
                \fill[fill=black!45] (j012) circle (1.5pt);
                \fill[fill=black!45] (j020) circle (1.5pt);
                \fill[fill=black!45] (j021) circle (1.5pt);
                \fill[fill=black!45] (j022) circle (1.5pt);
                \fill[fill=black!45] (j100) circle (1.5pt);
                \fill[fill=black!45] (j101) circle (1.5pt);
                \fill[fill=black!45] (j102) circle (1.5pt);
                \fill[fill=black!45] (j110) circle (1.5pt);
                \fill[fill=black!45] (j111) circle (1.5pt);
                \fill[fill=black!45] (j112) circle (1.5pt);
                \fill[fill=black!45] (j120) circle (1.5pt);
                \fill[fill=black!45] (j121) circle (1.5pt);
                \fill[fill=black!45] (j122) circle (1.5pt);
                \fill[fill=black!45] (j200) circle (1.5pt);
                \fill[fill=black!45] (j201) circle (1.5pt);
                \fill[fill=black!45] (j202) circle (1.5pt);
                \fill[fill=black!45] (j210) circle (1.5pt);
                \fill[fill=black!45] (j211) circle (1.5pt);
                \fill[fill=black!45] (j212) circle (1.5pt);
                \fill[fill=black!45] (j220) circle (1.5pt);
                \fill[fill=black!45] (j221) circle (1.5pt);
                \fill[fill=black!45] (j222) circle (1.5pt);
                \fill[fill=blue!70!black] (i) circle (2.4pt);
                \fill[fill=blue!70!black] (j) circle (2.4pt);
                \node[below=2pt] at (i) {\scriptsize $i$};
                \node[below=2pt] at (j) {\scriptsize $j$};
            \end{tikzpicture}%
        }
    \end{minipage}
    \caption{\label{figure:TreeExamples}
        The tree $T_{d}^{p}$ is built from a centre edge $i j$ (\textcolor{red!70!black}{red}) by including every vertex within graph distance $p$ of $i$ or $j$ in a $d$-regular graph $G$ of girth $\geq 2 (p + 1)$.
        Edges and vertices are shaded by their graph distance to the centre edge, from saturated blue at distance $1$, i.e., $T_{d}^{1}$, to faded grey at distance $3$; the latter lie outside $T_{d}^{2}$, and the dashed stubs indicate that $G$ extends even further.
        Shown are $d = 3$ on the left, and $d = 4$ on the right.
    }
\end{figure*}
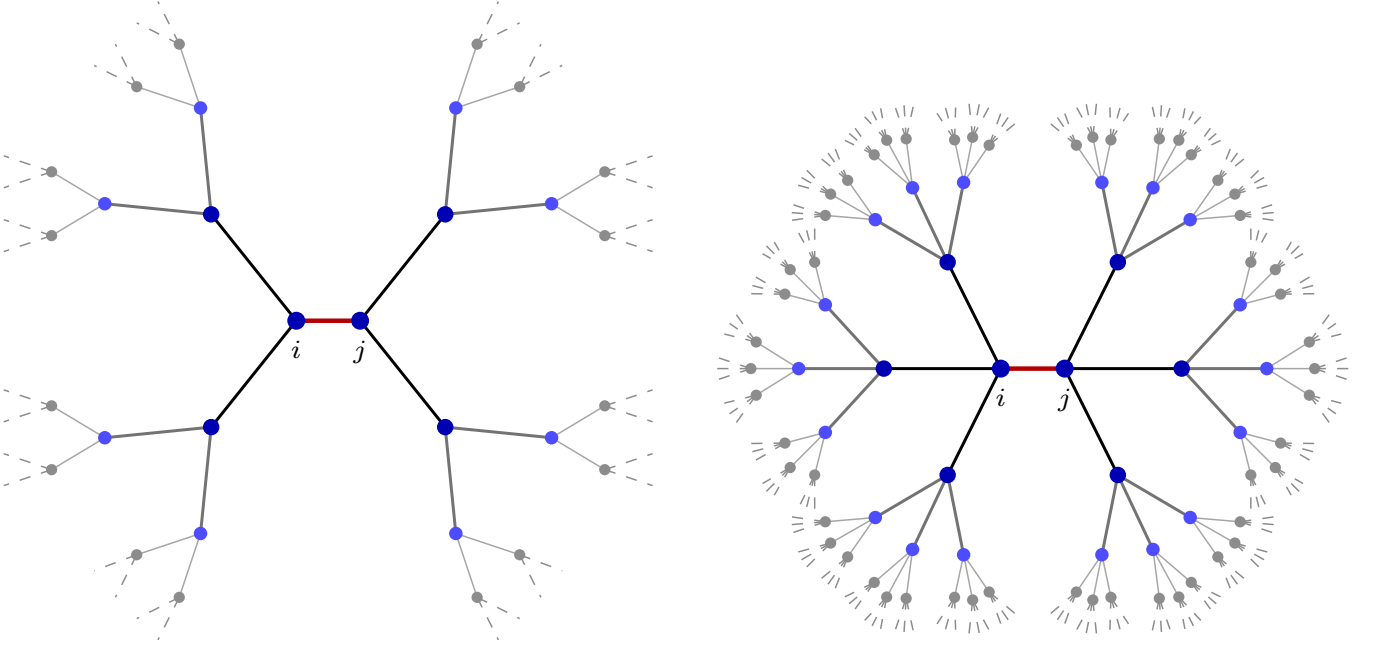

For large-girth graphs it is
\begin{align*}
    F(\bm{\beta}, \bm{\gamma}) = \abs{E} f^{\mathrm{tree}}_{p, d}(\bm{\beta}, \bm{\gamma}).
\end{align*}
As an immediate consequence, QAOA yields identical expectation values on all graphs with girth $\geq 2 (p + 1)$ and the same degree.
Furthermore, the parameter optimisation for the $n$-qubit system can be carried out using a single system of $\abs{T_{d}^{p}} = 2 \sum_{k = 0}^{p} (d - 1)^{k}$ qubits.

The precise value of $F(\bm{\beta}, \bm{\gamma})$ is often not that relevant, but its value relative to the optimal value is.
This quantity is called the \emph{approximation ratio}, defined as
\begin{align*}
    C(\bm{\beta}, \bm{\gamma}) \coloneqq \frac{F(\bm{\beta}, \bm{\gamma})}{F_{\max}}, \text{ where } F_{\max} \coloneqq \max_{\bm{x} \scriptin \bit^{n}} f(\bm{x}).
\end{align*}
Knowing the exact value of $F_{\max}$ evidently means that the MaxCut instance at hand is already solved.
However, this normalisation allows meaningful comparison of performance across different MaxCut instances.
One case where $F_{\max}$ is known precisely is when $G$ is also bipartite, i.e., when there exists a partition of its vertex set $V = V_{1} \mathbin{\mathaccent\cdot\cup} V_{2}$ so that every edge has one endpoint in $V_{1}$ and one in $V_{2}$.
Then, the optimal value for MaxCut is readily achieved by the partition into $V_{1}$ and $V_{2}$ itself, causing all edges of $G$ to be cut.
The subsequent corollary combines all three structural insights for $d$-regular, large-girth, bipartite graphs.

\begin{corollary}\label{corollary:WorstCaseGraph}
    For given $p, d \in \N$, let $G$ be a bipartite, $d$-regular graph with girth $\geq 2 (p + 1)$.
    Then, for every $(\bm{\beta}, \bm{\gamma}) \in \R^{2 p}$, it holds that
    \begin{align*}
        C(\bm{\beta}, \bm{\gamma}) = f_{p, d}^{\mathrm{tree}}(\bm{\beta}, \bm{\gamma}).
    \end{align*}
\end{corollary}

Throughout this article, the optimisers and optimal value of $f_{p, d}^{\mathrm{tree}}$ will be of utmost importance, so they deserve their own definition.

\begin{definition}
    For every $p, d \in \N$, let $(\bm{\beta}_{p, d}^{\mathrm{tree}}, \bm{\gamma}_{p, d}^{\mathrm{tree}}) \in \R^{2 p}$ be chosen such that
    \begin{align*}
        f_{p, d}^{\mathrm{tree}, *} \coloneqq f_{p, d}^{\mathrm{tree}}(\bm{\beta}_{p, d}^{\mathrm{tree}}, \bm{\gamma}_{p, d}^{\mathrm{tree}}) = \max_{(\bm{\beta}, \bm{\gamma}) \scriptin \R^{2 p}} f_{p, d}^{\mathrm{tree}}(\bm{\beta}, \bm{\gamma}).
    \end{align*}
\end{definition}

The combination of large-scale graph structure, to which the QAOA ansatz is blind, and the maximally high optimal value of $\abs{E}$ makes bipartite, large-girth graphs natural worst-case candidates for the QAOA, meaning the QAOA is expected to perform worse on these graphs than on every other comparable graph.
In this setting, the comparison runs over all $d$-regular graphs with the same $d$.
This is exactly the large-loop conjecture in~\cite{Wurtz2021FixedAngleConjecturesForTheQuantumApproximateOptimizationAlgorithmOnRegularMaxcutGraphs}.
Note that this is a statement comparing performances of QAOA with different parameter values, each optimal for their specific graph instance.
If true, this would give rigorous performance guarantees for the QAOA on any $d$-regular graph, solely by optimising $f_{p, d}^{\mathrm{tree}}$.
However, it would not help find the optimal (or even good) parameter values for the generic $d$-regular graph.

\begin{conjecture}\label{conjecture:LargeLoop}
    Let $p, d \in \N$, $G$ be a $d$-regular graph, and $(\bm{\beta}^{*}, \bm{\gamma}^{*})$ be a maximiser of $F$.
    Then, $C(\bm{\beta}^{*}, \bm{\gamma}^{*}) \geq f_{p, d}^{\mathrm{tree}, *}$ with equality if and only if $G$ is bipartite and has girth $\geq 2 (p + 1)$.
\end{conjecture}

A much stronger and practically relevant conjecture is that the inequality in \autoref{conjecture:LargeLoop} remains valid when reusing the tree-optimal parameters $\bm{\beta}_{p, d}^{\mathrm{tree}}$, $\bm{\gamma}_{p, d}^{\mathrm{tree}}$ instead of a graph's optimal parameters.
This is precisely the fixed-angle conjecture in~\cite{Wurtz2021FixedAngleConjecturesForTheQuantumApproximateOptimizationAlgorithmOnRegularMaxcutGraphs}.

\begin{conjecture}\label{conjecture:FixedAngle}
    Let $p, d \in \N$ and $G$ be a $d$-regular graph.
    Then, $C(\bm{\beta}_{p, d}^{\mathrm{tree}}, \bm{\gamma}_{p, d}^{\mathrm{tree}}) \geq f_{p, d}^{\mathrm{tree}, *}$ with equality if and only if $G$ is bipartite and has girth $\geq 2 (p + 1)$.
\end{conjecture}

This conjecture (without the bipartite requirement on $G$) is grounded in the following special case, proven in~\cite{Wurtz2021MaxcutQuantumApproximateOptimizationAlgorithmPerformanceGuaranteesFor}.

\begin{proposition}[$p \leq 2$; $d = 3$, Wurtz--Love \cite{Wurtz2021MaxcutQuantumApproximateOptimizationAlgorithmPerformanceGuaranteesFor}]\label{proposition:ThreeRegularFixedAngle}
    For $p~\leq~2$ and all $3$-regular graphs $G$, it is $C(\bm{\beta}_{p, 3}^{\mathrm{tree}}, \bm{\gamma}_{p, 3}^{\mathrm{tree}})~\geq~f_{p, 3}^{\mathrm{tree}, *}$ with equality if and only if $G$ is bipartite and has girth $\geq 2 (p + 1)$.
\end{proposition}

The subsequent analysis will occasionally be easier to conduct with the difference of edge expectation values and in terms of uncut edges, rather than cut edges.
Accordingly, let me define for an edge $e \in E$ of a $d$-regular graph the quantity $\delta_{e} \coloneqq f_{p, d}^{\mathrm{tree}, *} - f_{e}(\bm{\beta}_{p, d}^{\mathrm{tree}}, \bm{\gamma}_{p, d}^{\mathrm{tree}})$.
For the graph $G$ itself, let $\nu \coloneqq \abs{E} - F_{\max}$ be the number of edges that any maximum cut leaves uncut.
Then,
\begin{align*}
    C(\bm{\beta}_{p, d}^{\mathrm{tree}}, \bm{\gamma}_{p, d}^{\mathrm{tree}}) \geq f_{p, d}^{\mathrm{tree}, *} \iff \sum_{e \scriptin E} \delta_{e} \leq f_{p, d}^{\mathrm{tree}, *} \nu
\end{align*}
with equality on the left-hand side if and only if the right-hand side is an equality.

\section{\label{section:ValidationOfMoreSpecialCases}Validation of more special cases}

In this section, I am extending \autoref{proposition:ThreeRegularFixedAngle} to more $(p, d)$-pairs for which the fixed-angle conjecture holds true, before refuting the general validity in \autoref{section:ConcreteCounterexample}.
The first result establishes the statement for $p = 1$ and every $d \in \N$.
It directly utilises the analytic expression for $f_{i j}(\beta, \gamma)$ as well as the optimising parameter values for $f_{1, d}^{\mathrm{tree}}(\beta, \gamma)$ found by Wang~\textit{et~al.}~\cite{Wang2018QuantumApproximateOptimizationAlgorithmForMaxcutAFermionicView}.
Besides having a concrete analytical expression, taking $p = 1$ further simplifies the analysis, since bipartite graphs automatically have girth $\geq 4$, i.e., are triangle-free.
In a nutshell, the following proof utilises the simple form of $\delta_{e}$ for the $p = 1$ case, employs fundamental binomial identities to establish monotonic behaviour of $\delta_{e}$ in the number of triangles containing edge $e$, and combines it with other graph-specific inequalities into a final upper bound of $\sum_{e \scriptin E} \delta_{e}$ by $f_{1, d}^{\mathrm{tree}, *} \nu$, which becomes tight exactly in the bipartite, hence triangle-free case.

\begin{theorem}[$p = 1$; all $d$]\label{theorem:POneFixedAngle}
    For all $d \in \N$ and $d$-regular graphs $G$, it is $C(\bm{\beta}_{1, d}^{\mathrm{tree}}, \bm{\gamma}_{1, d}^{\mathrm{tree}}) \geq f_{1, d}^{\mathrm{tree}, *}$ with equality if and only if $G$ is bipartite and has girth $\geq 4$.
\end{theorem}

\begin{proof}
    Let $d \in \N$ be arbitrary, $G$ be a graph, and $i j \in E$ be an edge between two degree-$d$ vertices $i, j \in V$.
    Then, by \cite[Theorem 1]{Wang2018QuantumApproximateOptimizationAlgorithmForMaxcutAFermionicView}, it is
    \begin{align*}
        f_{i j}(\beta, \gamma) &= \frac{1}{2} \big(1 + \sin(4 \beta) \sin(\gamma) \cos^{d - 1}(\gamma)\big) \\
        &\quad - \frac{1}{4} \sin(2 \beta)^{2} \cos^{2 (d - \lambda - 1)}(\gamma)\big(1 - \cos^{\lambda}(2 \gamma)\big)
    \end{align*}
    for all $(\beta, \gamma) \in \R^{2}$, where $\lambda \coloneqq \abs{N(i) \cap N(j)}$ is the number of triangles in $G$ containing $i j$.
    In particular, for the triangle-free $T_{d}^{1}$, i.e., where $\lambda = 0$, it is
    \begin{align}\label{equation:OptimalTreeValue}
        f_{1, d}^{\mathrm{tree}}(\beta, \gamma) = \frac{1}{2} \big(1 + \sin(4 \beta) \sin(\gamma) \cos^{d - 1}(\gamma)\big).
    \end{align}
    Note that every $1$-regular graph is trivially bipartite and has girth $= \infty$ so that, by \autoref{corollary:WorstCaseGraph}, there is nothing to show for $d = 1$.
    Therefore, in the following, let $d \geq 2$.
    Following \cite[Corollary 1]{Wang2018QuantumApproximateOptimizationAlgorithmForMaxcutAFermionicView}, one pair of globally optimal parameters is $(\beta_{1, d}^{\mathrm{tree}}, \gamma_{1, d}^{\mathrm{tree}}) = (\pi / 8, \arctan(1 / \sqrt{d - 1}))$, producing the optimal value
    \begin{align}\label{equation:TreeOptimalInequality}
        f_{1, d}^{\mathrm{tree}, *} = \frac{1}{2} + \frac{(d - 1)^{(d - 1) / 2}}{2 d^{d / 2}} > \frac{1}{2}.
    \end{align}
    Now, let $G$ be any $d$-regular graph.
    Note that the right-hand side of~\eqref{equation:OptimalTreeValue} appears for all $f_{e}(\beta_{1, d}^{\mathrm{tree}}, \gamma_{1, d}^{\mathrm{tree}})$, $e \in E$, so that
    \begin{align*}
        \delta_{e}(\lambda) &= \frac{\sin^{2}(2 \beta_{1, d}^{\mathrm{tree}})}{4} \cos^{2 (d - \lambda - 1)}(\gamma_{1, d}^{\mathrm{tree}}) \big(1 - \cos^{\lambda}(2 \gamma_{1, d}^{\mathrm{tree}})\big) \\
        &= \tfrac{1}{8} \left(\tfrac{d - 1}{d}\right)^{d - \lambda - 1} \left(1 - \left(\tfrac{d - 2}{d}\right)^{\lambda}\right).
    \end{align*}
    Substitute $r \coloneqq (d - 1) / d$ and $\varepsilon \coloneqq 1 / (d - 1) = 1 / (d r)$, then
    \begin{align*}
        \delta_{e}(\lambda) &= \tfrac{r^{d - 1}}{8} \big((1 + \varepsilon)^{\lambda} - (1 - \varepsilon)^{\lambda}\big) \\
        &= \frac{r^{d - 1}}{4} \sum_{k = 1}^{\infty} \binom{\lambda}{2 k - 1} \varepsilon^{2 k - 1} \\
        &= \frac{r^{d - 1} \lambda}{4} \sum_{k = 1}^{\infty} \binom{\lambda - 1}{2 k - 2} \frac{\varepsilon^{2 k - 1}}{2 k - 1}
    \end{align*}
    by the binomial theorem.
    By Pascal's identity, it is
    \begin{align*}
        \binom{\lambda - 1}{2 k - 2} = \binom{\lambda - 2}{2 k - 2} + \binom{\lambda - 2}{2 k - 3} \geq \binom{(\lambda - 1) - 1}{2 k - 2},
    \end{align*}
    hence $\delta_{e}(\lambda)$ and, more importantly, $\delta_{e}(\lambda) / \lambda$ are non-decreasing in $\lambda \geq 1$.
    Additionally, note that $\lambda \leq d - 1$ so that, in summary,
    \begin{align*}
        &\delta_{e}(\lambda) / \lambda \leq \delta_{e}(d - 1) / (d - 1) \\
        &\implies \delta_{e}(\lambda) \leq c(d) \lambda \coloneqq \frac{\lambda}{8 (d - 1)} \left(1 - \left(\tfrac{d - 2}{d}\right)^{d - 1}\right)
    \end{align*}
    for all $\lambda \geq 1$.
    For $\lambda = 0$, it is simply $\delta_{e}(0) = 0 = c(d) \cdot 0$.
    Let $t$ be the number of triangles in $G$, and let $U \subset E$ be the uncut set of any maximum cut of $G$, i.e., it is $\abs{U} = \nu$.
    Every triangle contains at least one edge in $U$ and every edge $e \in E$ belongs to at most $d - 1$ triangles, so that, in total, $t \leq (d - 1) \nu$.
    Combining all these inequalities yields
    \begin{align*}
        &\sum_{e \scriptin E} \delta_{e} \leq c(d) \sum_{e \scriptin E} \lambda_{e} = 3 c(d) t \\
        &\leq 3 c(d) (d - 1) \nu = \tfrac{3}{8} \left(1 - \left(\tfrac{d - 2}{d}\right)^{d - 1}\right) \nu \leq \tfrac{3}{8} \nu \leq f_{1, d}^{\mathrm{tree}, *} \nu.
    \end{align*}
    If $\nu > 0$, the last inequality is strict by~\eqref{equation:TreeOptimalInequality}.
    In contrast, $\nu = 0$ if and only if $G$ is bipartite and, as a bipartite graph, has no triangles, i.e., has girth $\geq 4$, exactly establishing the equality case.
\end{proof}

Recent work by Marwaha~\cite{Marwaha2026TheQaoaOnTheRingOfDisagrees} has dealt extensively with the $d = 2$ case for arbitrary $p$, and builds on the free-fermionic representation established by Wang~\textit{et~al.}~\cite{Wang2018QuantumApproximateOptimizationAlgorithmForMaxcutAFermionicView}.
Marwaha was able to affirmatively resolve a conjecture already posed by Farhi~\textit{et~al.}~\cite{Farhi2014AQuantumApproximateOptimizationAlgorithm} about QAOA's performance on connected $2$-regular graphs, i.e., on rings.
The main theorems are therefore about the approximation ratio achieved by the graph-specific optimal parameters, but the machinery behind it is also powerful enough to fully settle the fixed-angle conjecture for $2$-regular graphs.
I will not reproduce all details of the construction and only recall its most important cornerstones.
In essence, Wang~\textit{et~al.} showed that running a single QAOA instance on a ring is equivalent to simultaneously running a QAOA instance on every other qubit of the ring.
The ring's structure, i.e., its length and its parity, translates to qubit-specific mode angles within the single-qubit objective Hamiltonians.
Therefore, the global QAOA's expectation value $F(\bm{\beta}, \bm{\gamma})$ becomes a sum of single-qubit QAOA expectation values $F_{j}(\bm{\beta}, \bm{\gamma}; \theta_{j})$ with respect to different, angle-dependent objective Hamiltonians.
Marwaha additionally identifies the parameters $(\bm{\beta}, \bm{\gamma})$ with Laurent polynomials, i.e., (complex) polynomials in $z$ and $z^{-1}$, by casting the single-qubit QAOA instances into the language of quantum signal processing.
Most importantly, the parameter-specific Laurent polynomial does not depend on the underlying graph structure so that the very simple Laurent polynomials corresponding to tree-optimal parameters $(\bm{\beta}_{p, 2}^{\mathrm{tree}}, \bm{\gamma}_{p, 2}^{\mathrm{tree}})$ can be reused for general rings (and unions thereof).
Simple trigonometric identities, exploiting the graph-specific mode angles and the Laurent polynomials' simple structure, then suffice for lower-bounding the QAOA's approximation ratio by its performance on the tree.
The methods slightly differ for even-length rings, which are always bipartite, and odd-length rings, but are similar in spirit.

\begin{theorem}[all $p$; $d \leq 2$]\label{theorem:SmallDFixedAngle}
    For all $p~\in~\N$,~$d~\leq~2$,~and $d$-regular graphs $G$, it is $C(\bm{\beta}_{p, d}^{\mathrm{tree}}, \bm{\gamma}_{p, d}^{\mathrm{tree}}) \geq f_{p, d}^{\mathrm{tree}, *}$ with~equality if and only if $G$ is bipartite and has girth $\geq 2 (p + 1)$.
\end{theorem}

\begin{proof}
    For $d = 1$, the same reasoning as in the proof of \autoref{theorem:POneFixedAngle} applies:
    Every $1$-regular graph is necessarily bipartite and has girth $= \infty > 2 (p + 1)$, valid for all $p \in \N$.
    Hence, there is again nothing to show for $d = 1$.

    The case $d = 2$ can be handled using Marwaha's recent work.
    Since $T_{2}^{p}$ with its degree-$1$ vertices connected has girth $= 2 (p + 1)$ for all $p \in \N$, \cite[Theorem 1]{Marwaha2026TheQaoaOnTheRingOfDisagrees} and \cite[Fact 13]{Marwaha2026TheQaoaOnTheRingOfDisagrees} are applicable for determining the optimal expectation value of its centre edge.
    Due to locality, this value is the same whether the endpoints of $T_{2}^{p}$ are connected or not, thus yielding
    \begin{align*}
        f_{p, 2}^{\mathrm{tree}, *} = \frac{2 p + 1}{2 (p + 1)} = 1 - \frac{1}{2 (p + 1)}.
    \end{align*}
    Let $(\bm{\beta}_{p, 2}^{\mathrm{tree}}, \bm{\gamma}_{p, 2}^{\mathrm{tree}})$ be any pair of tree-optimal parameters and let
    \begin{align*}
        L(z) = \sum_{m = 0}^{2 p + 1} c_{2 p + 1 - 2 m} z^{2 p + 1 - 2m},\ \left\lvert\sum_{m = 0}^{2 p + 1} c_{2 p + 1 - 2 m}\right\rvert^{2} = 1,
    \end{align*}
    be the associated Laurent polynomial.
    Expressing the QAOA approximation ratio on $T_{2}^{p}$ (with its degree-$1$ vertices connected) via $L$ yields
    \begin{align*}
        &f_{p, 2}^{\mathrm{tree}, *} = \frac{2 p + 1}{2 (p + 1)} = 1 - \sum_{m = 0}^{2 p + 1} \abs{c_{2 p + 1 - 2m}}^{2} \\
        \implies &\sum_{m = 0}^{2 p + 1} \abs{c_{2 p + 1 - 2 m}}^{2} = \frac{1}{2 (p + 1)}.
    \end{align*}
    In summary, the coefficients of $L$ saturate the Cauchy--Schwarz inequality
    \begin{align*}
        1 = \left\lvert\sum_{m = 0}^{2 p + 1} c_{2 p + 1 - 2m}\right\rvert^{2} \leq 2 (p + 1) \sum_{m = 0}^{2 p + 1} \abs{c_{2 p + 1 - 2m}}^{2} = 1
    \end{align*}
    and must therefore be all equal to some $\alpha \in \C$ with $\abs{\alpha} = 1 / (2 (p + 1))$, i.e., $\alpha = e^{i \phi} / (2 (p + 1))$ for a $\phi \in \R$.
    This leaves $L$ as a finite geometric series which simplifies even further under the substitution $z = e^{i \theta / 2}$ to
    \begin{align}
        L(e^{i \theta / 2}) &= \frac{e^{i \phi}}{2 (p + 1)} \sum_{m = 0}^{2 p + 1} \left(e^{i \theta / 2}\right)^{2 p + 1 - 2m} \nonumber \\
        &= \frac{e^{i \phi}}{2 (p + 1)} \frac{e^{i \theta (p + 1)} - e^{-i \theta (p + 1)}}{e^{i \theta / 2} - e^{-i \theta / 2}} \nonumber \\
        &= \frac{e^{i \phi}}{2 (p + 1)} \frac{\sin(\theta (p + 1))}{\sin(\theta / 2)}. \label{equation:SimplifiedLaurentPolynomial}
    \end{align}

    Note that all $2$-regular graphs are disjoint unions of cycles, so that any maximum cut of the union is given by the maximum cuts of its cycle components.
    Therefore, it suffices to restrict to single cycles.
    Consider first a cycle $G$ of even length $2 k$, $k \geq 2$, and note that $G$ is necessarily bipartite and has thus a maximum cut of $\abs{E} = 2 k$.
    If $G$ has girth $\geq 2 (p + 1)$, i.e., if $k > p$, the situation is exactly as for the glued $T_{2}^{p}$, yielding the exact same approximation ratio.
    Following \cite[Claim 3]{Marwaha2026TheQaoaOnTheRingOfDisagrees} and the amplitude representation underlying \cite[Corollary 10]{Marwaha2026TheQaoaOnTheRingOfDisagrees}, if $k \leq p$, it still holds that
    \begin{align*}
        C(\bm{\beta}_{p, 2}^{\mathrm{tree}}, \bm{\gamma}_{p, 2}^{\mathrm{tree}}) = \frac{F(\bm{\beta}_{p, 2}^{\mathrm{tree}}, \bm{\gamma}_{p, 2}^{\mathrm{tree}})}{2 k} = 1 - \frac{1}{k} \sum_{j = 1}^{k} \abs{L(e^{i \theta_{j} / 2})}^{2}
    \end{align*}
    with $\theta_{j} = \pi (2 j - 1) / (2 k)$.
    Using the simplified form~\eqref{equation:SimplifiedLaurentPolynomial}, the bound $\sin^{2}(\theta_{j} (p + 1)) \leq 1$ for all $j \in [k]$, and substituting $(n, k)$ in \cite[Proposition 3 (IV)]{Allouche2022HumanAndAutomatedApproachesForFiniteTrigonometricSums} by $(2 k, j - 1)$ yields the estimate
    \begin{align*}
        C(\bm{\beta}_{p, 2}^{\mathrm{tree}}, \bm{\gamma}_{p, 2}^{\mathrm{tree}}) &\geq 1 - \frac{1}{4 (p + 1)^{2} k} \sum_{j = 1}^{k} \csc^{2}\left(\tfrac{(2 j - 1) \pi}{4 k}\right) \\
        &= 1 - \frac{2 k^{2}}{4 (p + 1)^{2} k} = 1 - \frac{k}{2 (p + 1)^{2}} \\
        &> 1 - \frac{1}{2 (p + 1)} = f_{p, 2}^{\mathrm{tree}, *}.
    \end{align*}
    Now, consider a cycle $G$ of odd length $2 k + 1$, $k \geq 1$.
    The analysis is mostly analogous to the even case, but with the difference that in this case the maximum cut of $G$ is $\abs{E} - 1 = 2 k$.
    For $G$ of girth $\geq 2 (p + 1)$, i.e., with $k > p$, it is
    \begin{align*}
        C(\bm{\beta}_{p, 2}^{\mathrm{tree}}, \bm{\gamma}_{p, 2}^{\mathrm{tree}}) > \frac{F(\bm{\beta}_{p, 2}^{\mathrm{tree}}, \bm{\gamma}_{p, 2}^{\mathrm{tree}})}{2 k + 1} = f_{p, 2}^{\mathrm{tree}, *}.
    \end{align*}
    If instead $k \leq p$, \cite[Corollary 16]{Marwaha2026TheQaoaOnTheRingOfDisagrees} and the amplitude representation underlying~\cite[Corollary 17]{Marwaha2026TheQaoaOnTheRingOfDisagrees} yield
    \begin{align*}
        C(\bm{\beta}_{p, 2}^{\mathrm{tree}}, \bm{\gamma}_{p, 2}^{\mathrm{tree}}) = \frac{F(\bm{\beta}_{p, 2}^{\mathrm{tree}}, \bm{\gamma}_{p, 2}^{\mathrm{tree}})}{2 k} = 1 - \frac{1}{k} \sum_{j = 1}^{k} \abs{L(e^{i \vartheta_{j} / 2})}^{2}
    \end{align*}
    with angles $\vartheta_{j} = 2 \pi j / (2 k + 1)$.
    The combination of~\eqref{equation:SimplifiedLaurentPolynomial} with the fact that $\sin^{2}(\vartheta_{j} (p + 1)) \leq 1$ and the substitution of $(n, k)$ in \cite[Proposition 3 (I)]{Allouche2022HumanAndAutomatedApproachesForFiniteTrigonometricSums} by $(2 k + 1, j - 1)$ yields
    \begin{align*}
        C(\bm{\beta}_{p, 2}^{\mathrm{tree}}, \bm{\gamma}_{p, 2}^{\mathrm{tree}}) &\geq 1 - \frac{1}{4 (p + 1)^{2} k} \sum_{j = 1}^{k} \csc^{2}\left(\tfrac{j \pi}{2 k + 1}\right) \\
        &= 1 - \frac{2 k (k + 1)}{12 (p + 1)^{2} k} = 1 - \frac{k + 1}{6 (p + 1)^{2}} \\
        &\geq 1 - \frac{1}{6 (p + 1)} > 1 - \frac{1}{2 (p + 1)} = f_{p, 2}^{\mathrm{tree}, *}.
    \end{align*}
    This completes the analysis for connected $2$-regular graphs.
    A general $2$-regular graph, i.e., a disjoint union of connected $2$-regular graphs, is bipartite and has girth $\geq 2 (p + 1)$ if and only if this is true for all its connected components, thus directly inheriting the strict inequality and equality case.
\end{proof}

\section{\label{section:ConcreteCounterexample}Concrete counterexample}

Both regimes, in which the fixed-angle conjecture is true, benefit from relatively simple, closed-form expressions for the edge-wise expectation values $f_{e}(\bm{\beta}, \bm{\gamma})$, developed by Wang~\textit{et~al.}~\cite{Wang2018QuantumApproximateOptimizationAlgorithmForMaxcutAFermionicView}.
Later work by Marwaha~\cite{Marwaha2021LocalClassicalMaxCutAlgorithmOutperformsQaoaOnHighGirthRegularGraphs} provides a closed-form expression for the expectation values for $p = 2$ QAOA on all $d$-regular graphs with girth $\geq 6 = 2 (2 + 1)$.
In particular, it can be used to calculate $f_{2, d}^{\mathrm{tree}}(\bm{\beta}, \bm{\gamma})$.\footnote{
    To be precise, the formula applies to an embedding of $T_{d}^{p}$ into a $d$-regular graph by suitably connecting $T_{d}^{p}$'s leaves.
    The QAOA's expectation value on the centre edge is not affected by the additional edges due to the QAOA's locality.
}
However, it will not directly help with evaluating expectation values of a potential counterexample to the fixed-angle conjecture, since any counterexample for $p = 2$ necessarily has girth $< 6 = 2 (2 + 1)$.
I will not state the lengthy formula here and just reference Marwaha's main theorem in~\cite{Marwaha2021LocalClassicalMaxCutAlgorithmOutperformsQaoaOnHighGirthRegularGraphs}, whenever it is used.
For the counterexample, I utilise a Dicke state simulator, detailed in \autoref{section:DickeStateSimulator}.

Before I present the concrete counterexample, let me first provide some intuition to foster an educated guess for graphs challenging the fixed-angle conjecture.
Again, the situation is best understood in terms of the per-edge differences $\delta_{e}$, $e \in E$, and the number of uncut edges $\nu$.
Recall that the fixed-angle conjecture fails for a given $d$-regular graph $G$ and depth $p \in \N$ if and only if
\begin{align*}
    \sum_{e \scriptin E} \delta_{e} > f_{p, d}^{\mathrm{tree}, *} \nu.
\end{align*}
Short odd cycles have the potential to worsen $f_{e}(\bm{\beta}_{p, d}^{\mathrm{tree}}, \bm{\gamma}_{p, d}^{\mathrm{tree}})$, but simultaneously increase $\nu$ so that achieving the above strict inequality becomes more challenging.
In the extreme case, every odd cycle contributes $+ 1$ to $\nu$, regardless of its actual length.
However, the shorter the odd cycle is, the fewer single-edge expectation values are affected by it.
Therefore, at small $p$, where only tiny odd cycles are visible to the QAOA, odd cycles will most likely not be responsible for the failure of the fixed-angle conjecture.
Note that this intuition stems directly from the proof of \autoref{theorem:POneFixedAngle}, where the damage dealt to the single-edge expectation values by triangles, i.e., the smallest odd cycles, is compensated by the increase in $\nu$.

In comparison, short even cycles also have the potential to worsen single-edge expectation values, while not affecting $\nu$.
If indeed $\nu = 0$---which is the case if and only if $G$ is bipartite, i.e., has no odd cycles---the (necessary and) sufficient condition for the failure of the fixed-angle conjecture becomes
\begin{align*}
    \sum_{e \scriptin E} \delta_{e} > 0 \iff f_{p, d}^{\mathrm{tree}, *} > \tfrac{1}{\abs{E}} \sum_{e \scriptin E} f_{e}(\bm{\beta}_{p, d}^{\mathrm{tree}}, \bm{\gamma}_{p, d}^{\mathrm{tree}}).
\end{align*}
That is, for bipartite graphs, single-edge expectation values $f_{e}(\bm{\beta}_{p, d}^{\mathrm{tree}}, \bm{\gamma}_{p, d}^{\mathrm{tree}})$ must, on average, be worse than the tree reference value $f_{p, d}^{\mathrm{tree}, *}$.
The maximally distorted case is when an edge is part of many small even cycles, visible to the QAOA.
This also explains why even cycles do not make the fixed-angle conjecture false on $2$-regular graphs, where an edge is part of at most one cycle.
The smallest even cycles are of length $4$, invisible to the QAOA at $p = 1$, but already visible at $p = 2$.
For a given $d$, the unique, connected, bipartite, $d$-regular graph which maximises the number of four-cycles per edge is the complete bipartite graph $K_{d, d}$, depicted in \autoref{figure:CompleteBipartiteGraph}.
Here, every edge lies in $(d - 1)^{2}$ four-cycles.
As already implied by \autoref{proposition:ThreeRegularFixedAngle} and \autoref{theorem:SmallDFixedAngle}, small values for $d$ do not destructively disturb the local structure enough.
In fact, preliminary numerical studies suggest that $K_{d, d}$ with $d < 9$ is not a valid counterexample to the fixed-angle conjecture.
At $d = 9$, the local structure seems complex enough for the first time to significantly affect $f_{e}(\bm{\beta}_{2, 9}^{\mathrm{tree}}, \bm{\gamma}_{2, 9}^{\mathrm{tree}})$.
However, it is a priori not clear whether this effect always has to be destructive.
In the case of $K_{9, 9}$ it was found to be indeed destructive, but I explicitly do not want to conjecture that this is also the case for every $d > 9$.
Another important aspect of $K_{d, d}$ is that, already at $p = 2$, QAOA sees the entire graph, since any two vertices $i, j \in K_{d, d}$ have a distance to each other of at most $2$; if $i$ and $j$ are in opposite parts, their distance is $1$, if they are in the same part, their distance is $2$.
In particular, this means that all subgraphs $G_{i j}^{p}$, $p > 1$, are the same and isomorphic to $K_{d, d}$ itself.
Therefore, it holds that
\begin{align}\label{equation:BipartiteExpVal}
    F(\bm{\beta}, \bm{\gamma}) = \sum_{e \scriptin E} f_{e}(\bm{\beta}, \bm{\gamma}) = \abs{E} f_{p, d}^{K}(\bm{\beta}, \bm{\gamma}),
\end{align}
where $f_{p, d}^{K}$ is the single-edge expectation value of an arbitrary edge in $K_{d, d}$.
Furthermore, the MaxCut objective function on $K_{d, d}$ has two $d$-element permutation symmetries, one per part, as well as the binary symmetry of exchanging all $d$ bit (values) belonging to one part with the ones belonging to the other part.
More abstractly, $K_{d, d}$'s graph automorphism group is given by $\Sym(d) \times \Sym(d) \rtimes \Z_{2}$.
Since the objective function faithfully encodes the entire graph structure, it inherits this symmetry.
The mixer Hamiltonian carries a strictly richer $\Sym(2 d)$ symmetry.
Thus, the QAOA circuit at least retains the weaker symmetry of the objective function/Hamiltonian.
Writing out the initial state in (the symmetrisation of) per-part symmetrised states, i.e., (symmetrised) products of two $d$-qubit Dicke states, makes it possible to fully exploit this symmetry property.
Thereby, instead of simulating a $2^{2 d}$-dimensional system, the application of the QAOA circuit can be perfectly simulated on a symmetrised system of dimension $(d + 1) (d + 2) / 2$.
Concrete details can be found in \autoref{section:DickeStateSimulator}.

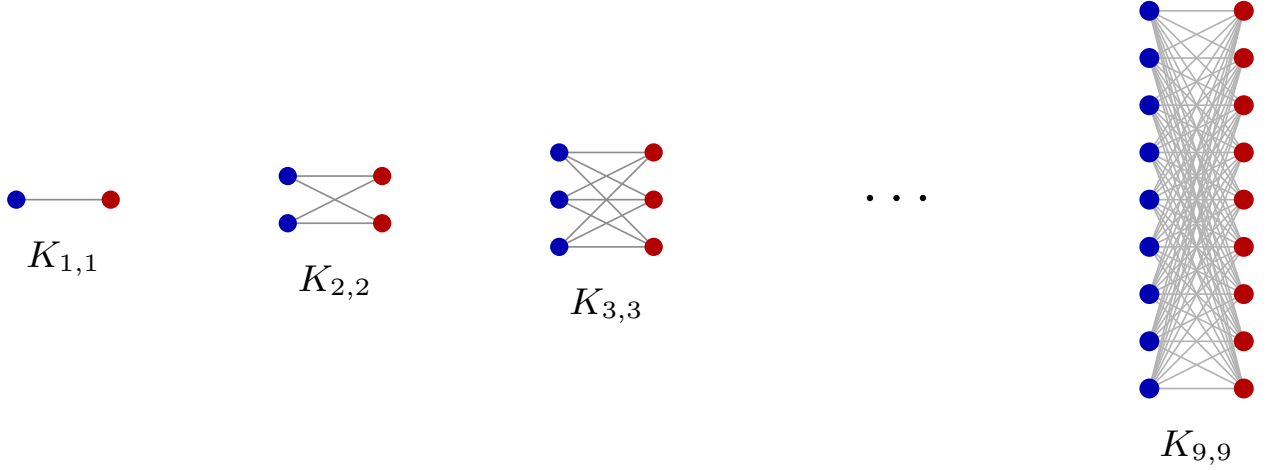
\begin{figure*}[t]
    \centering
    \resizebox{0.92\linewidth}{!}{%
        \begin{tikzpicture}[baseline]
            \begin{scope}[xshift=0cm]
                \foreach \i in {1} {
                    \pgfmathsetmacro{\y}{(\i - 1) * 0.4}
                    \coordinate (L\i) at (-0.4,\y);
                    \coordinate (R\i) at (0.4,\y);
                }
                \foreach \l in {1} {
                    \foreach \r in {1} {
                        \draw[black!45] (L\l) -- (R\r);
                    }
                }
                \foreach \i in {1} {
                    \fill[fill=blue!70!black] (L\i) circle (2.2pt);
                    \fill[fill=red!70!black] (R\i) circle (2.2pt);
                }
                \node at (0,-0.5) {\small $K_{1, 1}$};
            \end{scope}

            \begin{scope}[xshift=2.3cm]
                \foreach \i in {1,2} {
                    \pgfmathsetmacro{\y}{(\i - 1.5) * 0.4}
                    \coordinate (L\i) at (-0.4,\y);
                    \coordinate (R\i) at (0.4,\y);
                }
                \foreach \l in {1,2} {
                    \foreach \r in {1,2} {
                        \draw[black!45] (L\l) -- (R\r);
                    }
                }
                \foreach \i in {1,2} {
                    \fill[fill=blue!70!black] (L\i) circle (2.2pt);
                    \fill[fill=red!70!black] (R\i) circle (2.2pt);
                }
                \node at (0,-0.7) {\small $K_{2, 2}$};
            \end{scope}

            \begin{scope}[xshift=4.6cm]
                \foreach \i in {1,2,3} {
                    \pgfmathsetmacro{\y}{(\i - 2) * 0.4}
                    \coordinate (L\i) at (-0.4,\y);
                    \coordinate (R\i) at (0.4,\y);
                }
                \foreach \l in {1,2,3} {
                    \foreach \r in {1,2,3} {
                        \draw[black!45] (L\l) -- (R\r);
                    }
                }
                \foreach \i in {1,2,3} {
                    \fill[fill=blue!70!black] (L\i) circle (2.2pt);
                    \fill[fill=red!70!black] (R\i) circle (2.2pt);
                }
                \node at (0,-0.9) {\small $K_{3, 3}$};
            \end{scope}

            \node at (7.1,0) {\Large $\cdots$};

            \begin{scope}[xshift=9.6cm]
                \foreach \i in {1,...,9} {
                    \pgfmathsetmacro{\y}{(\i - 5) * 0.4}
                    \coordinate (L\i) at (-0.4,\y);
                    \coordinate (R\i) at (0.4,\y);
                }
                \foreach \l in {1,...,9} {
                    \foreach \r in {1,...,9} {
                        \draw[black!30, thin] (L\l) -- (R\r);
                    }
                }
                \foreach \i in {1,...,9} {
                    \fill[fill=blue!70!black] (L\i) circle (2.4pt);
                    \fill[fill=red!70!black] (R\i) circle (2.4pt);
                }
                \node at (0,-2.1) {\small $K_{9, 9}$};
            \end{scope}
        \end{tikzpicture}%
    }
    \caption{\label{figure:CompleteBipartiteGraph}
        The complete bipartite graphs $K_{d, d}$ for $d = 1, 2, 3, \dots, 9$; every vertex in one part is joined to every vertex in the other part, so each edge lies in $(d - 1)^{2}$ four-cycles.
        $K_{9, 9}$ serves as a counterexample to the fixed-angle conjecture for $p = 2$ on $9$-regular graphs.
    }
\end{figure*}

Let me now state the counterexample as a concrete theorem.
Note that directly proving that the tree-optimal angles perform worse on $K_{9, 9}$ than on $T_{9}^{2}$ would require precise knowledge of the true global optimisers of $f_{2, 9}^{\mathrm{tree}}$.
However, even using multi-start routines~\cite{Shaydulin2019MultistartMethodsForQuantumApproximateOptimization} as Wurtz and Lykov did~\cite{Wurtz2021FixedAngleConjecturesForTheQuantumApproximateOptimizationAlgorithmOnRegularMaxcutGraphs} cannot guarantee that a global optimum has been found.
While such techniques are perfectly fine for supporting numerical evidence, a strictly proven counterexample requires a (more) rigorous strategy.
Instead of trying to search for and rigorously certify that a given parameter point $(\bm{\beta}, \bm{\gamma})$ constitutes a global optimiser, I decided to test for a stronger but easier-to-verify statement, which then implies the original claim.
The statement consists of two parts.
First, a (presumably merely) local optimiser is used to obtain a threshold $b \in \R$ so that, by the very definition of the global optimum, $b \leq f_{2, 9}^{\mathrm{tree}, *}$.
Second, it is calculated that for all $(\bm{\beta}, \bm{\gamma}) \in \R^{4}$ with $f_{2, 9}^{\mathrm{tree}}(\bm{\beta}, \bm{\gamma}) \geq b$, $f_{2, 9}^{K}(\bm{\beta}, \bm{\gamma}) < f_{2, 9}^{\mathrm{tree}}(\bm{\beta}, \bm{\gamma})$.
Since the latter family of parameter values $(\bm{\beta}, \bm{\gamma})$ in particular contains any global optimiser, this is a strictly stronger statement.

\begin{theorem}\label{theorem:Counterexample}
    Let $G = K_{9, 9}$, $b \coloneqq 0.6391275408952286$, and $c \coloneqq 0.0029209162022571854$.
    Then,
    \begin{itemize}
        \item[i)] $b \leq f_{2, 9}^{\mathrm{tree}, *}$ and
        \item[ii)] for every $(\bm{\beta}, \bm{\gamma}) \in \R^{4}$ with $f_{2, 9}^{\mathrm{tree}}(\bm{\beta}, \bm{\gamma}) \geq b$ it is $f_{2, 9}^{K}(\bm{\beta}, \bm{\gamma}) - f_{2, 9}^{\mathrm{tree}}(\bm{\beta}, \bm{\gamma}) \leq -c$.
    \end{itemize}
\end{theorem}

\begin{proof}
    The proof itself is mostly computational.
    Here, I will only report on the high-level proof strategy; implementation details can be found in \autoref{section:ArbImplementationDetails}.
    The first part is to rigorously establish $b \leq f_{2, 9}^{\mathrm{tree}, *}$ under finite machine precision.
    Since $f_{2, 9}^{\mathrm{tree}, *}$ is, by definition, the optimal value for the tree expectation value, exactly evaluating $f_{2, 9}^{\mathrm{tree}}(\bm{\beta}', \bm{\gamma}')$ for any concrete $(\bm{\beta}', \bm{\gamma}') \in \R^{4}$ yields a rigorous lower bound.
    This evaluation is done via Marwaha's closed-form expression~\cite{Marwaha2021LocalClassicalMaxCutAlgorithmOutperformsQaoaOnHighGirthRegularGraphs} for $p = 2$ QAOA on $d$-regular graphs with girth $\geq 6$.
    However, ordinary floating-point number representations for determining $f_{2, 9}^{\mathrm{tree}}(\bm{\beta}', \bm{\gamma}')$ introduce potential sources for false-positive inequalities.
    Instead, I employ arbitrary-precision midpoint-radius interval arithmetic (Arb)~\cite{Johansson2017ArbEfficientArbitraryPrecisionMidpointRadiusIntervalArithmetic} to provably certify such inequalities.
    The Arb framework does not compute $f_{2, 9}^{\mathrm{tree}}(\bm{\beta}', \bm{\gamma}')$ exactly, but outputs an interval $[l, u]$ so that, provably, $f_{2, 9}^{\mathrm{tree}}(\bm{\beta}', \bm{\gamma}') \in [l, u]$.
    Then, the claim $b \leq f_{2, 9}^{\mathrm{tree}, *}$ reduces to verifying $b \leq l$.

    The second part of the proof is more involved, both conceptually and computationally.
    First, let me apply the usual symmetry considerations for the QAOA applied to MaxCut problems in order to reduce the unbounded parameter space to a bounded box.
    For the purely-quadratic MaxCut objective Hamiltonian, each $\beta$ parameter may be individually restricted to $[-\pi / 4, \pi / 4)$~\cite{Zhou2020QuantumApproximateOptimizationAlgorithmPerformanceMechanismAndImplementationOnNearTermDevices}.
    Additionally, since the objective Hamiltonian has integer spectrum, each $\gamma$ parameter may be individually restricted to $[-\pi, \pi)$.
    Furthermore, there is the time-reversal symmetry $f_{e}(\bm{\beta}, \bm{\gamma}) = f_{e}(-\bm{\beta}, -\bm{\gamma})$, which can be absorbed into any of the parameters; here I chose $\gamma_{1}$.
    These three symmetries condense the essential parameter space to $[-\pi / 4, \pi / 4) \times [0, \pi) \times [-\pi / 4, \pi / 4) \times [-\pi, \pi)$.
    The strategy is to cover the bounded parameter space with finitely many boxes
    \begin{align*}
        B_{i j k \ell} \coloneqq [l_{1}^{(i)}, u_{1}^{(i)}] \times [l_{2}^{(j)}, u_{2}^{(j)}] \times [l_{3}^{(k)}, u_{3}^{(k)}] \times [l_{4}^{(\ell)}, u_{4}^{(\ell)}],
    \end{align*}
    and to evaluate $f_{2, 9}^{\mathrm{tree}}$ for each of these boxes in Arb.
    The per-box result is an interval $[\tau_{l}, \tau_{u}] \coloneqq [\tau_{l}, \tau_{u}](i, j, k, \ell)$, provably containing all the values in $f_{2, 9}^{\mathrm{tree}}(B_{i j k \ell})$.
    If $\tau_{u}$ sits strictly below $b$, then no $(\bm{\beta}, \bm{\gamma}) \in B_{i j k \ell}$ can fulfil $f_{2, 9}^{\mathrm{tree}}(\bm{\beta}, \bm{\gamma}) \geq b$, disqualifying the entire box as being relevant for ii).
    Otherwise there might be parameter values $(\bm{\beta}, \bm{\gamma}) \in B_{i j k \ell}$ with $f_{2, 9}^{\mathrm{tree}}(\bm{\beta}, \bm{\gamma}) \geq b$ so that the second condition in ii) has to be checked as well.
    This is again done for the entire box at once, using Arb.
    The interval $[\tau_{l}, \tau_{u}]$, provably containing $f_{2, 9}^{\mathrm{tree}}(B_{i j k \ell})$, has already been computed for the first condition, so it remains to compute a second interval $[\kappa_{l}, \kappa_{u}] \coloneqq [\kappa_{l}, \kappa_{u}](i, j, k, \ell)$, provably containing $f_{2, 9}^{K}(B_{i j k \ell})$.
    This is done by combining Arb with a Dicke state simulator, detailed in \autoref{section:DickeStateSimulator}.
    Naively subtracting both intervals would already produce the interval $[\kappa_{l} - \tau_{u}, \kappa_{u} - \tau_{l}]$, provably containing $(f_{2, 9}^{K} - f_{2, 9}^{\mathrm{tree}})(B_{i j k \ell})$.
    Since this discards the strong correlation between both functions, the implementation instead encloses the difference directly via the mean-value form described in \autoref{section:ArbImplementationDetails}, yielding a tighter interval $[\delta_{l}, \delta_{u}] \supseteq (f_{2, 9}^{K} - f_{2, 9}^{\mathrm{tree}})(B_{i j k \ell})$.
    Then, if $\delta_{u} \leq -c$, ii) is established for the entire box.
    Note that a box $B_{i j k \ell}$ failing the first test might still contain many values $(\bm{\beta}, \bm{\gamma})$ with $f_{2, 9}^{\mathrm{tree}}(\bm{\beta}, \bm{\gamma}) < b$.
    Yet, the second test covers these values as well.
    Additionally, if the difference interval's upper bound $\delta_{u}$ is too loose, it might lie strictly above $-c$, even though all true values $(f_{2, 9}^{K} - f_{2, 9}^{\mathrm{tree}})(B_{i j k \ell})$ might lie below $-c$.
    Accordingly, a failure of the second test on a specific box does not automatically imply that the statement of \autoref{theorem:Counterexample} is wrong.
    Instead, it simply means that the box $B_{i j k \ell}$ was taken too large for any Arb-based test to be decisive.
    $B_{i j k \ell}$ is now recursively split into smaller boxes until ii) can be rigorously certified for all subboxes, covering $B_{i j k \ell}$.
    This follows the same logic as the initial coverage of the entire parameter space by boxes $B_{i j k \ell}$.
\end{proof}

\begin{corollary}[Counterexample for $p = 2$, $d = 9$]\label{corollary:Counterexample}
    For $G = K_{9, 9}$, it is $C(\bm{\beta}_{2, 9}^{\mathrm{tree}}, \bm{\gamma}_{2, 9}^{\mathrm{tree}}) < f_{2, 9}^{\mathrm{tree}, *}$.
\end{corollary}

\begin{proof}
    Let $(\bm{\beta}_{2, 9}^{\mathrm{tree}}, \bm{\gamma}_{2, 9}^{\mathrm{tree}}) \in \R^{4}$ be any maximiser of $f_{2, 9}^{\mathrm{tree}}$, i.e., it holds by \autoref{theorem:Counterexample} i) that
    \begin{align*}
        f_{2, 9}^{\mathrm{tree}}(\bm{\beta}_{2, 9}^{\mathrm{tree}}, \bm{\gamma}_{2, 9}^{\mathrm{tree}}) = f_{2, 9}^{\mathrm{tree}, *} \geq b.
    \end{align*}
    Then, by \autoref{theorem:Counterexample} ii) and~\eqref{equation:BipartiteExpVal}, it is
    \begin{align*}
        C(\bm{\beta}_{2, 9}^{\mathrm{tree}}, \bm{\gamma}_{2, 9}^{\mathrm{tree}}) &= \frac{1}{\abs{E}} \sum_{e \scriptin E} f_{e}(\bm{\beta}_{2, 9}^{\mathrm{tree}}, \bm{\gamma}_{2, 9}^{\mathrm{tree}}) \\
        &= f_{2, 9}^{K}(\bm{\beta}_{2, 9}^{\mathrm{tree}}, \bm{\gamma}_{2, 9}^{\mathrm{tree}}) \leq f_{2, 9}^{\mathrm{tree}, *} \hspace*{-2pt} - c < f_{2, 9}^{\mathrm{tree}, *}.
    \end{align*}
    This holds for all maximisers $(\bm{\beta}_{2, 9}^{\mathrm{tree}}, \bm{\gamma}_{2, 9}^{\mathrm{tree}})$ of $f_{2, 9}^{\mathrm{tree}}$.
\end{proof}

This concludes the rigorous refutation of the fixed-angle conjecture.
However, note that $K_{9, 9}$ is not (necessarily) a counterexample to \autoref{conjecture:LargeLoop}, which compares the performance of the QAOA with per-graph optimised parameters.
For the sake of completeness, let me state this fact as a separate result.

\begin{lemma}\label{lemma:NoCounterexample}
    Let $G = K_{9,9}$, $u \coloneqq 0.6391276594250365$, and
    \begin{align*}
        \bm{w} \coloneqq \begin{pmatrix}
            \beta_{1} \\
            \gamma_{1} \\
            \beta_{2} \\
            \gamma_{2}
        \end{pmatrix} \coloneqq \begin{pmatrix}
            0.5030008586984759 \\
            2.9153856978850885 \\
            -0.46951423323060537 \\
            2.8693708283343287
        \end{pmatrix}.
    \end{align*}
    Then,
    \begin{itemize}
        \item[i)] $f_{2, 9}^{\mathrm{tree}, *} \leq u$ and
        \item[ii)] $f_{2, 9}^{K}(\bm{w}) > u$.
    \end{itemize}
\end{lemma}

\begin{proof}
    The proof is again computational and builds directly on top of the proof of \autoref{theorem:Counterexample}, which already established $f_{2, 9}^{\mathrm{tree}}(\bm{\beta}, \bm{\gamma}) < b (< u)$ on a large portion of the parameter domain.
    Accordingly, only those boxes have to be evaluated on which $b$ could not be established as an upper bound.
    The rigorous box-wide inequality is again established using Arb.
    Additionally, the witness $\bm{w}$ is fed into the Arb-assisted Dicke state simulator, producing an interval $[\omega_{l}, \omega_{u}] \ni f_{2, 9}^{K}(\bm{w})$.
    The subsequent check confirms that $\omega_{l} > u$.
\end{proof}

\begin{corollary}[No counterexample to large-loop]\label{corollary:NoCounterexample}
    For $p = 2$ and $G = K_{9, 9}$, it holds that $C(\bm{\beta}^{*}, \bm{\gamma}^{*}) > f_{2, 9}^{\mathrm{tree}, *}$, where $(\bm{\beta}^{*}, \bm{\gamma}^{*}) \in \R^{4}$ are any maximisers of $F$ (on $G$).
\end{corollary}

\begin{proof}
    Let $\bm{w}$ be as in \autoref{lemma:NoCounterexample} and $(\bm{\beta}^{*}, \bm{\gamma}^{*})$ be maximisers of $F$ (on $K_{9, 9}$).
    Then, by \autoref{lemma:NoCounterexample} i) and ii) it is
    \begin{align*}
        C(\bm{\beta}^{*}, \bm{\gamma}^{*}) \geq f_{2, 9}^{K}(\bm{w}) > u \geq f_{2, 9}^{\mathrm{tree}, *}.
    \end{align*}
    Note that $K_{9, 9}$ has girth $= 4 < 6$, so that \autoref{conjecture:LargeLoop} predicts the above strict inequality.
\end{proof}

\section{\label{section:Conclusion}Conclusion}

In this article, I have presented a concrete counterexample at $p = 2$ and $d = 9$ for the fixed-angle conjecture by Wurtz and Love~\cite{Wurtz2021MaxcutQuantumApproximateOptimizationAlgorithmPerformanceGuaranteesFor} and its refinement by Wurtz and Lykov~\cite{Wurtz2021FixedAngleConjecturesForTheQuantumApproximateOptimizationAlgorithmOnRegularMaxcutGraphs}.
In a nutshell, the fixed-angle conjecture claims the existence of universally good parameter values for the depth-$p$ QAOA ansatz tackling MaxCut on $d$-regular graphs, supplying approximation ratios bounded below by the optimal QAOA performance on some worst-case graphs.
Past work has been mostly focused on the $3$-regular case and has consistently yielded positive results: proven for $p \leq 2$~\cite{Wurtz2021MaxcutQuantumApproximateOptimizationAlgorithmPerformanceGuaranteesFor} and numerically for $p \leq 11$~\cite{Wurtz2021FixedAngleConjecturesForTheQuantumApproximateOptimizationAlgorithmOnRegularMaxcutGraphs}.
The counterexample at $d = 9$ has no direct consequences for the $3$-regular route, which---if correct---promises guaranteed approximation ratios above the Goemans--Williamson bound at $p = 11$.
However, the counterexample presented here refutes the possibility of adapting this strategy verbatim to general $d$-regular graphs.
Furthermore, the counterexample to the fixed-angle conjecture does not refute the weaker large-loop conjecture, comparing the performance of the graph's own optimal parameters to the optimal performance on worst-case graphs.

While the fixed-angle conjecture in its most general form has been disproved, hope remains for low-$p$ and low-$d$ special cases.
I have presented proofs for two special cases, complementing the known $d = 3$ case for $p \leq 2$:
First, for $p = 1$, a combination of elementary binomial identities and the analytic expression found by Wang~\textit{et~al.}~\cite{Wang2018QuantumApproximateOptimizationAlgorithmForMaxcutAFermionicView} for single-edge QAOA expectation values proves that the $p = 1$ version of the conjecture indeed holds true for all $d \in \N$.
Second, by combining elementary trigonometric identities, Wang~\textit{et~al.}'s closed-form expression for single-edge QAOA expectation values on connected $2$-regular graphs, and the connection between QAOA parameters and Laurent polynomials, recently established by Marwaha~\cite{Marwaha2026TheQaoaOnTheRingOfDisagrees}, I showed that the $d \leq 2$ version of the conjecture is true for all $p \in \N$.
However, MaxCut on $1$-regular and $2$-regular graphs is trivial, so that the strategy of increasing $p$---as for the $3$-regular case---does not imply any potential quantum advantage.

\begin{acknowledgments}
    This work was supported by the Federal Ministry for Research, Technology and Space (BMFTR) under grant number 13N17249.

    \noindent\textbf{Data and code availability statement.}
    All code supporting the findings of this study is available at \url{https://github.com/MarkAureli/wl-counterexample}.
\end{acknowledgments}

\section*{AI Disclosure}

I have used Claude Opus 5 and Claude Sonnet 5.
I had a long, interactive brainstorming session with Opus 5 about the conjecture and why it might not generalise to larger $d$, eventually converging to the insight that a potential counterexample would probably maximise the number of short even cycles per edge for the reasons stated in the article.
Opus 5 then immediately suggested inspecting $K_{d, d}$ for large enough $d$.
The initial numerical proof, implemented by Sonnet 5, was fundamentally flawed, as it only established strictly worse performance for one local optimiser.
The revised version then used the provable coverage of the entire parameter space via the Arb-based implementation, faithfully proving \autoref{theorem:Counterexample}.
During a discussion with the orchestrating Opus 5, ``we'' realised that the counterexample-producing mechanism is not available when $p = 1$ (no four-cycles are visible) or $d = 2$ (only one cycle per edge).
Opus 5 then quickly found the binomial and trigonometric identities which are used within the proofs of \autoref{theorem:POneFixedAngle} and \autoref{theorem:SmallDFixedAngle}.

I reviewed the entire code base, implemented by Sonnet 5, refactored small parts, and convinced myself that the implementation actually addresses the claim.
This entire document was written by me and was passed once through Opus 5 for grammar and notational consistency checks.

\bibliographystyle{apsrev4-2}
\bibliography{bibliography}

\onecolumngrid

\twocolumngrid

\appendix

\section{\label{section:DickeStateSimulator}Dicke state simulator}

In this section, I detail the Dicke state simulator employed in the numerical proof of \autoref{theorem:Counterexample}.
The observation that the QAOA circuit for the MaxCut problem on $K_{d, d}$ possesses a $\Sym(d) \times \Sym(d) \rtimes \Z_{2}$ symmetry suggests conducting the simulation within the basis of states obeying the same symmetry condition.
For simplicity, assume that the first $d$ qubits represent vertices in one part, so that the remaining $d$ qubits represent the vertices in the other part.
A state $\ket{\psi} \in \qubit^{\otimes 2d}$ is symmetric under $\Sym(d) \times \Sym(d) \rtimes \Z_{2}$ if and only if it is a superposition of states of the form
\begin{align*}
    \ket{p, q} \coloneqq \begin{cases}
        \ket{D_{p}^{d}} \ket{D_{q}^{d}}, & \text{if } p = q \\
        \tfrac{1}{\sqrt{2}} \left(\ket{D_{p}^{d}} \ket{D_{q}^{d}} + \ket{D_{q}^{d}} \ket{D_{p}^{d}}\right), & \text{else}, \\
    \end{cases}
\end{align*}
for some $0 \leq p \leq q \leq d$, where
\begin{align*}
    \ket{D_{p}^{d}} \coloneqq \binom{d}{p}^{-1 / 2} \sum_{\substack{\bm{x} \scriptin \bit^{d} \\ \abs{\bm{x}} = p}} \ket{\bm{x}}
\end{align*}
is a Dicke state on $d$ qubits.
Furthermore, note that $\braket{p, q | p', q'} = \delta_{p p'} \delta_{q q'}$ so that these
\begin{align*}
    1 + 2 + \ldots + (d + 1) = \frac{(d + 1) (d + 2)}{2}
\end{align*}
states constitute an orthonormal basis of a subspace of dimension $(d + 1) (d + 2) / 2$.
The QAOA's initial state $\ket{\bm{+}}$ already lies within this subspace and has the following representation:
\begin{align*}
    \ket{\bm{+}} = 2^{-d} \sum_{\bm{x} \scriptin \bit^{2 d}} \ket{\bm{x}} = \sum_{q = 0}^{d} \sum_{p = 0}^{q} \sqrt{\frac{(2 - \delta_{p q}) \binom{d}{p} \binom{d}{q}}{2^{2 d}}} \ket{p, q}.
\end{align*}
In particular, the entire application of the QAOA circuit to its initial state restricts to the symmetrised subspace.
It remains to derive the concrete form of $\exp(-i \gamma H_{f})$ and $\exp(-i \beta B)$ in the symmetrised basis.
On $K_{d, d}$ specifically, it is
\begin{align*}
    H_{f} = \sum_{i j \scriptin E} \tfrac{1}{2} (\one - \texttt{Z}_{i} \texttt{Z}_{j}) = \tfrac{d^{2}}{2} \one - \tfrac{1}{2} \left(\sum_{i = 1}^{d} \texttt{Z}_{i}\right) \left(\sum_{j = d + 1}^{2 d} \texttt{Z}_{j}\right).
\end{align*}
Since
\begin{align*}
    \sum_{i = 1}^{d} \texttt{Z}_{i} \ket{D_{p}^{d}} = \big(p \cdot (-1) + (d - p) \cdot 1\big) \hspace*{-1pt} \ket{D_{p}^{d}} = (d - 2 p) \hspace*{-1pt} \ket{D_{p}^{d}}
\end{align*}
and for the second tensor factor analogously, it follows that
\begin{align*}
    H_{f} \ket{D_{p}^{d}} \ket{D_{q}^{d}} &= \left(\tfrac{d^{2}}{2} - \tfrac{1}{2} (d - 2p) (d - 2 q)\right) \ket{D_{p}^{d}} \ket{D_{q}^{d}} \\
    &= (d (p + q) - 2 p q) \ket{D_{p}^{d}} \ket{D_{q}^{d}} \implies \\
    H_{f} \ket{p, q} &= (d (p + q) - 2 p q) \ket{p, q}
\end{align*}
for all $0 \leq p \leq q \leq d$.
In particular, $H_{f}$ remains diagonal in the symmetrised basis.
Its matrix exponential accordingly has the diagonal form
\begin{align*}
    e^{-i \gamma H_{f}} \ket{p, q} = e^{-i \gamma \left(d (p + q) - 2 p q\right)} \ket{p, q}.
\end{align*}

In contrast, the $2 d$-qubit mixer Hamiltonian $B^{(2 d)}$ will not be diagonal in the symmetrised basis, requiring the study of its matrix exponential directly.
Nevertheless, the mixer Hamiltonian already fulfils
\begin{align*}
    B^{(2 d)} = \sum_{i = 1}^{2 d} \texttt{X}_{i} = \sum_{i = 1}^{d} \texttt{X}_{i} +\hspace*{-3pt} \sum_{j = d + 1}^{2 d} \texttt{X}_{j} = B^{(d)} \otimes \one + \one \otimes B^{(d)}\hspace*{-1pt},
\end{align*}
so that $\exp(-i \beta B^{(2 d)}) = \exp(-i \beta B^{(d)}) \otimes \exp(-i \beta B^{(d)})$.
The $d$-qubit matrix $(\langle D_{p'}^{d} | \exp(-i \beta B^{(d)}) | D_{p}^{d}\rangle)_{p' p}$ is given by the Wigner (small) d-matrix~\cite[XV: (27)]{Wigner1931GruppentheorieUndIhreAnwendungAufDieQuantenmechanikDerAtomspektren} with angles $\alpha = \gamma = 0$ and $\beta \mapsto 2 \beta$, spin $j = d / 2$, and indices $\mu' = p' - d / 2$ and $\mu = p - d / 2$, and a per-element prefactor of $(-i)^{p' - p}$, i.e.,
\begin{widetext}
    \begin{align*}
        M_{p' p} \coloneqq \left\langle D_{p'}^{d} \left\lvert e^{-i \beta B^{(d)}} \right\rvert D_{p}^{d}\right\rangle = \sum_{k \scriptin \Z} \frac{\sqrt{p!\, (d - p)!\, p'!\, (d - p')!}}{(d - p' - k)!\, (p - k)!\, k!\, (k + p' - p)!} \big(\cos(\beta)\big)^{d + p - p' - 2 k} \big(-i \sin(\beta)\big)^{2 k + p' - p}.
    \end{align*}
\end{widetext}
The above sum runs from $k = \max\{0, p -p'\}$ to $k = \min\{d - p', p\}$.
The form of the matrix elements $\langle p', q' | e^{-i \beta B} | p, q\rangle$ now follows directly from the tensor product structure of $B = B^{(2 d)}$ and the definition of $\ket{p, q}$.
The simplest case is $p = q$ and $p' = q'$, where
\begin{align*}
    \left\langle p', p' \left\lvert e^{-i \beta B}\right\rvert p, p\right\rangle = (M_{p' p})^{2}.
\end{align*}
If instead $p \neq q$ and $p' \neq q'$, then
\begin{align*}
    \left\langle p', q' \left\lvert e^{-i \beta B}\right\rvert p, q\right\rangle = M_{p' p} M_{q' q} + M_{p' q} M_{q' p}.
\end{align*}
However, if $p \neq q$ but $p' = q'$, it holds that
\begin{align*}
    \left\langle p', p' \left\lvert e^{-i \beta B}\right\rvert p, q\right\rangle = \sqrt{2} M_{p' p} M_{p' q}.
\end{align*}
Analogously, if $p = q$ but $p' \neq q'$, then
\begin{align*}
    \left\langle p', q' \left\lvert e^{-i \beta B}\right\rvert p, p\right\rangle = \sqrt{2} M_{p' p} M_{q' p}.
\end{align*}

\section{\label{section:ArbImplementationDetails}Arb implementation details}

In Arb~\cite{Johansson2017ArbEfficientArbitraryPrecisionMidpointRadiusIntervalArithmetic}, a real number $x \in \R$ is always enclosed in an interval $[m \pm r]$, where $m$ is an arbitrary-precision float and $r$ is an unsigned float with arbitrary exponent range (and fixed mantissa).
If $x$ is exactly expressible in the working precision $p$, $x = m$ and $r = 0$.
Otherwise, $m$ is set to the correctly-rounded value of $x$ in the target precision $p$, $\round_{p}(x)$, and $r > 0$ is set to the upward-rounded resulting rounding error, ensuring $x \in [m \pm r]$.
Basic arithmetic operations and transcendental functions, including $\sin$, $\cos$, and $\exp$, use analytic error bounds to construct a ball $[\tilde{m} \pm \tilde{r}]$, provably containing the function values of all points in the input ball(s).
For example, adding two Arb balls $[m_{1} \pm r_{1}]$ and $[m_{2} \pm r_{2}]$ at precision $p$ produces the midpoint $\tilde{m} = \round_{p}(m_{1} + m_{2})$ and radius $\tilde{r} \geq r_{1} + r_{2} + \abs{\round_{p}(m_{1} + m_{2}) - (m_{1} + m_{2})}$, provably containing all values $x_{1} + x_{2}$ with $x_{i} \in [m_{i} \pm r_{i}]$, $i = 1, 2$.
For transcendental functions, determining tight radii is more involved, but essentially boils down to repeatedly applying arithmetic operations for evaluating Taylor polynomials with rigorous remainder bounds.

Concatenating Arb primitives successively inflates the ball radius after each step.
This is problematic for long calculations, such as those needed for the Dicke state-based simulation.
An alternative approach, already used within Arb itself, is to infer $\tilde{r}$ from the mean value theorem:
For $f : \R \rightarrow \R$ differentiable, it is
\begin{align*}
    \sup_{\abs{t} \leq r} \abs{f(m + t) - f(m)} \leq r \sup_{\abs{t} \leq r} \abs{f'(m + t)},
\end{align*}
which gives a valid $\tilde{r}$ upon upward-rounding.
This requires, in turn, an upper estimate of $\abs{f'}$ on the ball $[m \pm r]$, which can be provided by the ``ordinary'' Arb method for $f'$.
This naturally extends to differentiable functions $f : \R^{n} \rightarrow \R$, simply by applying the mean value theorem in each coordinate direction.
For $f : \R^{n} \rightarrow \R$, it is
\begin{align*}
    \sup_{\forall i : \abs{t_{i}} \leq r_{i}} \abs{f(m_{1} + t_{1}, \ldots, m_{n} + t_{n}) - f(m_{1}, \ldots, m_{n})} \\
    \leq \sum_{j = 1}^{n} r_{j} \sup_{\forall i : \abs{t_{i}} \leq r_{i}} \abs{\partial_{j} f(m_{1} + t_{1}, \ldots, m_{n} + t_{n})},
\end{align*}
requiring upper bounds of $\abs{\partial_{j} f}$ for all $j \in [n]$ on the $n$-dimensional box $[m_{1} \pm r_{1}] \times \cdots \times [m_{n} \pm r_{n}]$.
In the implementation, these gradient enclosures are obtained by forward-mode automatic differentiation over Arb balls, i.e., by propagating value and all four partial derivatives through the closed-form tree expression and the Dicke state simulator simultaneously.

The covering in the proof of \autoref{theorem:Counterexample} is organised as a binary box tree, processed depth-first and starting from the symmetry-reduced domain, whose endpoints are rounded outward to the neighbouring binary64 numbers.
All computations run at a working precision of $96$ bits.
Every box is enclosed by a ball $[m_{1} \pm r_{1}] \times \cdots \times [m_{4} \pm r_{4}]$ with binary64 midpoints and upward-rounded radii, and every bound extracted from a ball is rounded outward once more.
Each box is subjected to up to three tests, in this order:
\begin{itemize}
    \item[P\,n)] natural Arb evaluation of $f_{2, 9}^{\mathrm{tree}}$ on the ball has supremum $< b$;
    \item[P\,c)] if the natural supremum exceeds $b$ by at most $0.03$, the centred (mean-value) form of $f_{2, 9}^{\mathrm{tree}}$ has supremum $< b$;
    \item[S)] if, additionally, every side of the box has width at most $1.5 \times 10^{-4}$, the centred form of the difference $f_{2, 9}^{K} - f_{2, 9}^{\mathrm{tree}}$, with gradient enclosure of the difference, has supremum $< 0$.
\end{itemize}
Leaves of type P\,n or P\,c contain no parameters with $f_{2, 9}^{\mathrm{tree}} \geq b$, and leaves of type S satisfy ii) with their certified supremum.
A box passing none of the applicable tests is bisected at the exact binary64 midpoint of the coordinate with the largest first-order contribution $r_{i} \sup \abs{\partial_{i} f_{2, 9}^{\mathrm{tree}}}$ (or the widest coordinate if no gradient was computed), so that the leaves tile the domain exactly.
The constant $c$ in \autoref{theorem:Counterexample} is the negative of the largest certified supremum over all S leaves.

The final certificate is a depth-$62$ box tree of $3\,237\,923$ nodes, consisting of $1\,618\,961$ internal nodes and $1\,618\,962$ leaves.
Of the leaves, $298\,679$ are discharged by P\,n, $1\,307\,259$ by P\,c, and $13\,024$ by S.
All S leaves have a largest side width of $\pi / 2^{15} \approx 9.59 \times 10^{-5}$.
In total, the run required $5\,146\,656$ ball evaluations of $f_{2, 9}^{\mathrm{tree}}$, $1\,908\,733$ gradient enclosures of $f_{2, 9}^{\mathrm{tree}}$, and $13\,024$ gradient enclosures of $f_{2, 9}^{K}$, taking about $410\,\mathrm{s}$ on a single core via the python-flint bindings of Arb.
The lower bound in \autoref{theorem:Counterexample} i) is established at $256$ bits, where the witness value exceeds $b$ by $1.58 \times 10^{-16}$.

The certificate is re-verified by an independent checker, which uses the inf-sup interval arithmetic of mpmath~\cite{TheMpmathDevelopmentTeam2023MpmathAPythonLibraryForArbitraryPrecisionFloatingPointArithmeticVersion130} at $25$ significant decimal digits.
The checker reconstructs every box from the recorded split coordinates, confirms that the leaves exactly tile the reduced domain, verifies the containment of each box in its enclosing ball in exact rational arithmetic, and re-establishes the classification of every single leaf.
For the S leaves, it uses the slightly looser bound $\sup \abs{\partial_{i} f_{2, 9}^{K}} + \sup \abs{\partial_{i} f_{2, 9}^{\mathrm{tree}}}$ on the partial derivatives of the difference, and recomputes the violation margin from scratch rather than reading it from the certificate.
Since the checker's bound is looser, this margin may be marginally smaller than the value stated in \autoref{theorem:Counterexample}; the constant $c$ stated there is the one obtained from the Arb run, and the checker's role is to independently confirm that every S leaf satisfies $\sup (f_{2, 9}^{K} - f_{2, 9}^{\mathrm{tree}}) < 0$.
The witness inequality $b \leq f_{2, 9}^{\mathrm{tree}}(\bm{\beta}', \bm{\gamma}')$ is re-checked at $40$ digits.
The full replay takes about two hours.

The proof of \autoref{lemma:NoCounterexample} i) reuses the same box tree.
P leaves already satisfy $\sup f_{2, 9}^{\mathrm{tree}} < b \leq u$, so only the $13\,024$ S leaves are re-evaluated with the centred form of $f_{2, 9}^{\mathrm{tree}}$.
The upper bound $u$ is taken to be the largest of these certified suprema.
For ii), the Dicke state simulator at $256$ bits yields $f_{2, 9}^{K}(\bm{w}) \geq 0.79287$, exceeding $u$ by more than $0.15$.

\end{document}